\documentclass[a4paper,10pt]{article}
\usepackage[utf8]{inputenc}
\usepackage{amssymb, MnSymbol, bbm}
\usepackage{amsmath, amsthm}
\usepackage{epsfig, enumerate}
\usepackage{color}

\let\OLDthebibliography\thebibliography
\renewcommand\thebibliography[1]{
  \OLDthebibliography{#1}
  \setlength{\parskip}{0pt}
  \setlength{\itemsep}{0.2pt plus 0.3ex}
}

\newcommand{\bphi}{\boldsymbol{\varphi}}

\newcommand{\bpsi}{\boldsymbol{\Psi}}

\newtheorem{theorem}{Theorem}
\newtheorem{lemma}[theorem]{Lemma}
\newtheorem{coro}[theorem]{Corollary}
\newtheorem{prop}[theorem]{Proposition}

\newtheorem{conjecture}{Conjecture}

\newtheorem*{assumptions*}{Assumptions}

\newtheorem*{rem}{Remark}
\newtheorem*{question}{Question}
\newtheorem*{problem}{Problem}
\newtheorem{remark}[theorem]{Remark} 

\theoremstyle{definition}
\newtheorem{defini}{Definition}

\newcommand{\bfs}{{\mathbf s}}

\newcommand{\Aa}{{\mathcal A}}
\newcommand{\Bb}{{\mathcal B}}

\newcommand{\Dd}{{\mathcal D}}

\newcommand{\Gg}{{\mathcal G}}
\newcommand{\Hh}{{\mathcal H}}
\newcommand{\Ii}{{\mathcal I}}
\newcommand{\Jj}{{\mathcal J}}

\newcommand{\Oo}{{\mathcal O}}
\newcommand{\Pp}{{\mathcal P}}

\newcommand{\Rr}{{\mathcal R}}

\newcommand{\Vv}{{\mathcal V}}

\newcommand{\Ab}{{\mathbb A}}

\newcommand{\CC}{{\mathbb C}}

\newcommand{\EE}{{\mathbb E}}

\newcommand{\GG}{{\mathbb G}}
\newcommand{\HH}{{\mathbb H}}

\newcommand{\NN}{{\mathbb N}}

\newcommand{\PP}{{\mathbb P}}
\newcommand{\QQ}{{\mathbb Q}}
\newcommand{\RR}{{\mathbb R}}

\newcommand{\VV}{{\mathbb V}}

\newcommand{\ZZ}{{\mathbb Z}}

\newcommand{\Sss}{{\mathfrak S}}

\newcommand{\sss}{{\mathfrak s}}

\newcommand{\one}{{\bf 1}}
\newcommand{\nul}{{\bf 0}}

\newcommand{\qtx}[1]{\quad\text{#1}\quad}

\newcommand{\wt}{\widetilde}

\newcommand{\Mat}{{\rm Mat}}

\newcommand{\Her}{{\rm Her}}

\newcommand{\pmat}[1]{\begin{pmatrix} #1  \end{pmatrix}}
\newcommand{\smat}[1]{\left( \begin{smallmatrix} #1  \end{smallmatrix} \right)}

\newcommand{\diag}{{\rm diag}}
\DeclareMathOperator{\im}{{\rm Im}}
\DeclareMathOperator{\re}{{\rm Re}}

\DeclareMathOperator{\ran}{{\rm Ran}}

\DeclareMathOperator{\Var}{{\rm Var}}
\DeclareMathOperator{\supp}{{\rm supp}}
\DeclareMathOperator{\spec}{{\rm spec}}
\DeclareMathOperator{\sgn}{{\rm sgn}}

\newcommand{\ess}{{\rm ess}}

\numberwithin{theorem}{section}
\numberwithin{equation}{section}

\title{Anderson transition at 2 dimensional growth rate on antitrees and spectral theory for operators with one propagating channel}

\author{Christian Sadel\thanks{The research of C.S. has received funding from the People Programme (Marie Curie Actions) of the European 
Union's Seventh Framework Programme (FP7/2007-2013) under REA grant agreement number 291734.}\\[.2cm]
{\small Institute of Science and Technology, Austria}}

\begin{document}
\maketitle 

\begin{abstract}
 We show that the Anderson model has a transition from localization to delocalization at exactly 2 dimensional growth rate on antitrees with normalized edge weights which are certain discrete graphs. The kinetic part has a one-dimensional structure allowing a description through transfer matrices which involve some Schur complement. 
 For such operators we introduce the notion of having one propagating channel and extend theorems from the theory of one-dimensional Jacobi operators that relate the behavior of transfer matrices with the spectrum. These theorems are then applied to the considered model. In essence, in a certain energy region the kinetic part averages the random potentials along shells and the transfer matrices behave similar as for a one-dimensional operator with random potential of decaying variance.
 At $d$ dimensional growth for $d>2$ this effective decay is strong enough to obtain absolutely continuous spectrum, whereas for some uniform $d$ dimensional growth 
 with $d<2$ one has pure point spectrum in this energy region. At exactly uniform $2$ dimensional growth also some singular continuous spectrum appears, at least at small disorder.
 As a corollary we also obtain a change from singular spectrum ($d\leq 2$) to absolutely continuous spectrum ($d\geq 3)$ 
 for random operators of the type $\Pp_r \Delta_d \Pp_r+\lambda \Vv$ on $\ZZ^d$, where $\Pp_r$ is an orthogonal radial projection, $\Delta_d$ the discrete adjacency operator (Laplacian) on $\ZZ^d$ and $\lambda \Vv$ a random potential.
\end{abstract}

\tableofcontents

\section{Introduction}

Anderson models are random Schr\"odinger operators given by the sum of a kinetic operator and a random potential.
In a discrete setting the kinetic part is typically the adjacency operator or Laplacian of a discrete graph and the random potential is a multiplication operator with real, independent, identically distributed random values at each vertex. 
Most commonly studied are the Anderson models on the lattices $\ZZ^d$ and the Bethe lattices (infinite regular trees).
In these cases and for continuous versions consisting of the negative Laplacian and random potential bumps 
in $\RR^d$ several things are known. The Anderson model typically localizes (has pure point spectrum) at spectral edges and for high disorder
\cite{FS,FMSS,DLS,SW,CKM,DK,Kl1,AM,A,Wang,Klo, BK}. However, so far, the high disorder localization in the discrete setup requires some regularity on the randomness, localization for the Bernoulli potential in $\ZZ^d$, $d\geq 2$ is still an open problem. 
In one dimension \cite{GMP,KuS,CKM} and quasi-one dimensional graphs like trees with long line sequences \cite{Breu} and strips \cite{Lac,KlLS} the Anderson model localizes for any disorder.
But it is possible that a built in symmetry prevents localization for a quasi-one dimensional random operator as e.g. in \cite{SS}. 
For $d=2$ one expects localization at any disorder and for $d\geq 3$ the existence of some absolutely continuous spectrum (short a.c. spectrum) is expected for small disorder. These conjectures remain big open problems.

The existence of a.c. spectrum for the Anderson model has first been proved on Bethe lattices (regular trees) \cite{Kl3} and was extended to other
tree-like graphs with exponentially growing boundary which are all infinite dimensional
\cite{ASW, FHS1, FHS2, H, KLW, KLW2, FHH, KS, AW, Sad, Sad2, Sha}.
It appears that the hyperbolic nature of such graphs leads to conservation of
a.c. spectrum and ballistic dynamical behavior for small disorder \cite{Kl2, KS2, AW2}.
Using the fractional moment method \cite{AM} one finds localization at high disorder also on all these graphs (cf. \cite{Tau}).
We therefore have some Anderson transition (change in spectral behavior) when increasing the disorder.

Transitions of the spectral type are also known for random decaying potentials on $\ZZ^d$ when changing the decay rate \cite{KLS, Bou} and similar families of random 
Jacobi operators \cite{BL, BFS}.
Even (at least) two transitions occur when increasing a random transversally periodic potential which is added to a random radial symmetric potential of fixed disorder on a binary tree, cf.
\cite[Corollary~1.5]{FLSSS}. Here we seek for a transition when changing the dimension. In fact, these results and open conjectures raise the following questions.
\begin{question}
 Can one find some finite dimensional graph $\GG$ (in the sense of a polynomial growth of the graph) such that the Anderson model at small disorder with i.i.d. potential has absolutely continuous spectrum and such that
 the analysis is simpler than the difficult problem on $\ZZ^d$, $d\geq 2$?
 Can one find a 'nice' family of graphs $\GG_d$ of dimension $d\geq 1$ such that the spectral type of the Anderson model changes
 at some $d\in[2,3]$ for low disorder?
\end{question}

\subsection{Dimension of discrete graphs}

Let $\GG$ be some graph with countably many vertices which is edge connected and locally 
finite\footnote{This means that at each vertex the vertex degree, i.e. number of edges at that vertex, is finite. Note that we do not assume a uniform upper bound}.
The graph distance, or better step distance $d(x,y)\in\ZZ_+$ is defined by the smallest number of steps needed to go from $x$ to $y$ along edges,
i.e. the smallest number $n$ such that there is a sequence $x=x_0,\, x_1,\ldots,x_n=y$ where
$x_i$ and $x_{i+1}$ are connected by an edge. We define further $d(x,x)=0$ and as $\GG$ is edge connected we have $d(x,y)<\infty$ for any $x,y\in\GG$. 
Let us choose some non-empty, finite set of vertices $S_0\subset\GG$ which we call the roots of $\GG$ and let
\begin{gather*}
S_n :=\{x\in\GG\,:\, d(S_0,x)=n\},\quad \Bb_n:=\{x\in\GG\,:\, d(S_0,x)\leq n\},\,\\
\Sss_n:=\{x\in\Bb_n\,:\,d(x,S_{n+1})=1\,\}\;,\quad s_n:=\#(S_n),\;b_n:=\#(\Bb_n),\;\sss_n:=\#(\Sss_n) \;
\end{gather*}
Clearly, $\Bb_n=\bigcup_{j=0}^n S_j,\;b_n=\sum_{j=0}^n s_j$ and $0<s_j<\infty$ as $\GG$ is connected, has infinitely many vertices and is locally finite.
$S_n$ is the shell or sphere of distance $n$ around $S_0$.
The boundary of a set $\Lambda\subset\GG$ is given by $\partial \Lambda=\{x\,:\,d(x,\Lambda)\leq 1\;\text{and}\;d(x,\GG\setminus\Lambda)\leq 1\}$.
One has an interior boundary $\partial \Lambda\cap \Lambda$ and an exterior boundary $\partial \Lambda \setminus \Lambda$.
The interior boundary of $\Bb_n$ is $\Sss_n$ and the exterior boundary is $S_{n+1}$.

\begin{defini} Let $d\geq 1$ be some real number.\\
 (a) We say that the volume growth of $\GG$ is
 \begin{itemize}
 \item  $d$-dimensional, if $cn^{d} < b_n < Cn^{d}$ for $C>c>0$ and all $n$.
 \item at least (at most) $d$-dimensional, if $b_n>Cn^d$ (resp. $b_n<Cn^d$) for some $C>0$ and all $n$.
 \end{itemize}
 (b) We say that (starting from $S_0$) the growth rate of $\GG$ is
 \begin{itemize}
 \item  $d$-dimensional, if $cn^{d-1} < s_n < Cn^{d-1}$ for $C>c>0$ and all $n$.
 \item {\bf uniform} $d$-dimensional, if $\,\lim\limits_{n\to\infty} s_n\,/\,n^{d-1}\,=\,C$ for $C>0$.
 \item at least (at most) $d$-dimensional, if $s_n>Cn^{d-1}$ (resp. $s_n<Cn^{d-1}$) for $C>0$ and all $n$.
 \end{itemize}
 (c) The logarithmic ratio of the interior (resp. exterior) surface-area and volume has a $d$-dimensional behavior if 
 $\lim\limits_{n\to\infty} \log(\sss_n)\,/\,{\log(b_n)}\,=\,1-\frac1d\;$, resp. 
 $\lim\limits_{n\to\infty} {\log(s_{n+1})}\,/\,{\log(b_n)}\,=\,1-\frac1d\;$
\end{defini}
Except for the constant, the bounds on the volume growth do not depend on the choice of the roots $S_0$. 
Clearly, the growth rate bounds on $s_n$ imply the corresponding volume growth bounds. A $d$-dimensional growth rate also implies a $d$-dimensional behavior for the logarithmic ratio of the exterior surface-area and volume.

The more restrictive concept of a uniform $d$-dimensional growth-rate will be needed for some of the presented results. Let us mention that the lattice $\ZZ^d$ has a uniform $d$-dimensional growth rate for any $d$ and any finite set of roots $S_0$ (cf. Section~\ref{sec:PDP}.)

The above concepts do not give much information about the structure of edges. If the vertex degree is not uniformly bounded one may have strange effects. For graphs with such a uniform bound the adjacency and Laplace operators are bounded.
Therefore, in the general case one should allow for normalizing edge weights, i.e. we assign some $0\neq a(x,y)=a(y,x) \in \RR$ whenever $d(x,y)=1$ and let $a(x,y)=0$ if $d(x,y)\neq 1$.
The weights shall be chosen such that the matrix $a(x,y)$ defines a bounded, self-adjoint operator $\Aa$ on $\ell^2(\GG)$ with some delocalized energy region, i.e.
there should be a non-empty open set $I\subset \spec_{{\rm ac}}(\Aa)$ such that the spectrum of $\Aa$ is purely a.c. in $I$. 
Let us denote the class of such weighted graphs $(\GG,a)$ by $\Gg$.

The Anderson model on $(\GG,a)$ is then given by the sum of $\Aa$ with a random, independent identically distributed potential $\Vv$ coupled with some $\lambda>0$ determining the disorder, i.e.
\begin{equation}\label{eq-op}
(H_\lambda \psi)(x) \,=\, ((\Aa \,+\,\lambda \Vv)\psi)\,(x)\;=\;\sum_{y:d(x,y)=1} a(x,y)\psi(y)\,+\,\lambda \,v(x)\psi(x)\;,
\end{equation}
where the $v(x)$, $x\in\GG$, are real, independent, identically distributed random variables. 
Typically $\EE(v(x))=0$ and $\EE(v^2(x))<\infty$ where $\EE$ denotes the expectation value. 
\begin{problem}
Find 'nice' families of graphs $(\GG_d,a_d) \in\Gg$ with $d$-dimensional growth rates for any $d$ 
where the spectral type of the Anderson model changes at some $d\in[2,3]$ in some interval.
\end{problem}
\noindent We give such an example in this work.

\subsection{Main result}

The main objects are the following graphs which (with standard weights $a(x,y)=1$) are called
antitrees in \cite{Woj, KLWo, BrK}.

\begin{defini}
For a sequence $\bfs = (s_n)_{n\geq 0}$ of positive integers $s_n>0$ we let $\Ab_\bfs$ be the following graph:
The $n$-th shell $S_n$ consists of $s_n$ vertices, each vertex in $S_n$ is connected to each vertex in $S_{n\pm1}$ and for $x\in S_n, \, y\in S_{n+1}$ the edge
from $x$ to $y$ obtains the weight $a(x,y)=1/\sqrt{s_n s_{n+1}}$. There are no edges within $S_n$ (see Figure~\ref{fig:1}).
We  call $\Ab_\bfs$ with these weights the {\bf antitree with normalized edge-weights associated to the sequence $\bfs=(s_n)_{n\geq0}$}.
The corresponding weighted adjacency operator given by the matrix $a(x,y)$ will be called $\Aa_\bfs$. (Recall, $a(x,y)=0$ if there is no edge from $x$ to $y$, also note that as above, $S_n=\{y:d(y,S_0)=n\}$.)
\end{defini}

\begin{figure}[ht]
\begin{center}
 \includegraphics[width=5cm]{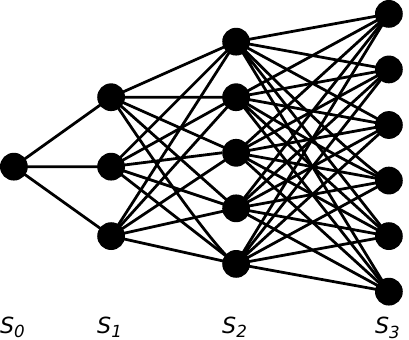}
\end{center}
\caption{Example for an antitree
}\label{fig:1}
\end{figure}

Breuer and Keller \cite{BrK} calculated the spectrum of the Laplacian and adjacency operator on antitrees with standard weights ($a(x,y)=1$ whenever $a(x,y)\neq 0$).
They considered general so-called spherically homogeneous graphs where it is feasible to use the Gram-Schmidt orthonormalization procedure on a sequence 
$\psi,\, H\psi,\,H^2\psi\,\ldots$ to get a Jacobi operator.
Compared to their paper the presented methods are different and the methods work in principle for any potential and additional hopping terms within the shells $S_n$
destroying the spherical homogeneity.

The normalization of the edge weights in this way ensures that $\Aa_\bfs$ is bounded on $\ell^2(\Ab_\bfs)$ and the
spectrum is $[-2,2]$ independent of the sequence $\bfs=(s_n)_n$ which corresponds to some remark in \cite{BrK}.
\begin{prop}\label{prop-H0}
 The spectrum of the weighted adjacency operator $\Aa_\bfs$ on the antitree with normalized edge-weights $\Ab_\bfs$ is given 
 by $\sigma(\Aa_\bfs)=[-2,2]$ and it is absolutely continuous except for a (possibly) embedded eigenvalue at $0$ with multiplicity $\sum_{n=0}^\infty (s_n-1)$ (which may be infinite). 
\end{prop}

For a better understanding of the normalization consider the ordinary adjacency operator $\Delta_d$ on $\ZZ^d$, its spectrum is $[-2d,2d]$.
The graph distance on $\ZZ^d$ is given by the $1$-norm $\|\cdot\|_1$, hence let $s_n=\#\{x\in\ZZ^d\,:\,\|x\|_1=n\}$.
Now let $\Pp_r$ be the orthogonal projection on the radial subspace $\HH_r=\{\psi\in\ell^2(\ZZ^d):\psi(x)=\psi(y) \text{ whenever }\|x\|_1=\|y\|_1\}$. Then, $\frac1d \Pp_r\Delta_d \Pp_r$ is basically the adjacency operator of an antitree $\Ab_\bfs$ with weights
$a(x,y)=\left(1+\Oo(1/n^2)\right)\,/\,\sqrt{s_n s_{n+1}}$ if $\|x\|_1=n,\, \|y\|_{1}=n+1$ (cf. Corollary~\ref{coro:main} and Section~\ref{sec:PDP}).
Hence, asymptotically $\frac1d \Pp_r\Delta_d \Pp_r$ looks like $\Aa_\bfs$.

Basically we will show the following: Consider the Anderson model $H_\lambda$ on the set of antitrees $\Ab_\bfs$ with uniform $d$-dimensional growth rate
and certain nice distributions $\PP_v$ of the singe site potentials $v(x)$.
Then, the spectral type of $H_\lambda$ in some set $I_\lambda$ is pure point for $d<2$, it is mixed pure point and singular continuous at $d=2$ and it is purely absolutely continuous for $d>2$.

\begin{assumptions*} $\phantom{H}$
\begin{enumerate}
 \item[{\rm (A1)}] The distribution $\PP_v$ of the i.i.d. random potentials $v(x)$ is supported in $[-1,1]$, has mean zero $\EE_v(v)=0$ 
 and positive variance $0<\sigma^2:=\EE_v(v^2)\leq 1$\,.
 \item[{\rm (A2)}] The distribution $\PP_v$ of the i.i.d. random potentials $v(x)$ is absolutely continuous with respect to the Lebesgue measure.
\end{enumerate}
\end{assumptions*}
\noindent Here and below, $\EE_v(f(v))=\int f(v) \PP_v(dv)$.
Assumption (A1) will be always important, (A2) will matter for the singular spectrum.
We let $\supp(\PP_v)$ denote the support of the measure $\PP_v$. Under assumption (A1) we can define \mbox{$0<v_{+}\leq 1$} and $-1\leq v_{-}<0$ 
by 
$$ 
 v_{+}\,=\,\max(\supp(\PP_v))\qtx{and}  v_{-}\,=\,\min(\supp(\PP_v))\;.
$$
For $E/\lambda \not\in\supp(\PP_v)$ we define
\begin{equation}\label{eq-def-h}
h_{E,\lambda}:= \frac{1}{\EE_v(1/(E-\lambda v))} \qtx{and set} I_\lambda:=\{E\in\RR\,:\, E\not\in[\lambda v_-,\lambda v_+]\;\;\text{and}\;\;|h_{E,\lambda}|<2\;\}\,.
\end{equation}
Note that $h_{E,\lambda}$ could be technically infinite if $\EE_v(1/(E-\lambda v))=0$ but this will not happen for $E\in I_\lambda$. It is not difficult to see that
$I_\lambda=(c'_{\lambda},\lambda v_-) \cup (\lambda v_+, c_\lambda)$, however for large $\lambda$ one or both of these intervals may be empty (see also Remark~\ref{rem:I_l}).
For concrete examples see Remark~\ref{rem:I_l}.
Moreover, for $E\in I_\lambda$ we define
\begin{equation}\label{eq-def-gamma}
\sigma_{E,\lambda}^2:=\EE_v\left(\frac{1}{(E-\lambda v)}-h_{E,\lambda}^{-1}\right)^2\;,\qquad
 \gamma_{E,\lambda}\,:=\,\frac{h_{E,\lambda}^4\,\sigma^2_{E,\lambda}}{2(4-h_{E,\lambda}^2)}\,>\,0\,.
\end{equation}
To get a feeling of these quantities for small disorder, note that for $\lambda\to 0$ we have 
\begin{equation}\label{eq-scaling}
h_{E,\lambda}\,=\,E+\Oo(\lambda^2)\,,\qquad \sigma_{E,\lambda}^2\,=\,\Oo(\lambda^2)\,.
\end{equation}
Last but not least, let us also introduce the canonical injection $P_n$ from $\ell^2(S_n)\cong\CC^{s_n}$ into $\ell^2(\Ab_\bfs)$ so that for $\psi\in\ell^2(\Ab_\bfs)$ one has 
$P_n^*\psi=(\psi(x))_{x\in S_n}\in\ell^2(S_n)\cong \CC^{s_n}$. One can identify $\psi$ with the direct sum $\bigoplus_n P_n^* \psi \in \bigoplus_n \ell^2(S_n)\cong \bigoplus_n \CC^{s_n}$.
The main result is the following:

\begin{theorem}\label{th:main} Let $H_\lambda=\Aa_\bfs+\lambda\Vv$ be the Anderson model on the antitree $\Ab_\bfs$ as in \eqref{eq-op} 
with i.i.d. potential satisfying {\rm (A1)} and let $\lambda>0$. 
Then, $I_\lambda$ is not empty for $\lambda\leq 2+\frac{2\sigma^2}{2-\sigma^2}$ and for the graphs $\Ab_\bfs$ with uniform $d$-dimensional growth rate there is a transition of the spectral type in $I_\lambda$ at $d=2$.
More precisely, we find the following:
\begin{enumerate}[{\rm (i)}]
 \item If $\Ab_\bfs$ has at least $d$ dimensional growth rate for some $d>2$ (i.e. $s_n>cn^{d-1}$), and in fact, whenever $\sum_n s_n^{-1}<\infty$, then, the spectrum of $H_\lambda$ is almost surely purely
 absolutely continuous in $I_\lambda$ and $I_\lambda$ is in the spectrum.
 \item If $\Ab_\bfs$ has uniform $d$-dimensional growth rate for some $1\leq d\leq 2$ ($s_n/n^{d-1}$ converges), then the spectrum of $H_\lambda$ in $I_\lambda$ is almost surely singular and $I_\lambda$ is in the spectrum. 
 \item If additionally {\rm (A2)} is satisfied and $\Ab_\bfs$ has uniform $d$-dimensional growth rate for some $1<d< 2$ with
 $\lim_{n\to\infty} s_n/n^{d-1}= C$, then, $H_\lambda$ has almost surely dense pure point spectrum in $I_\lambda$.
 Moreover, almost surely the random eigenvectors $\psi_j$ for the random eigenvalues $E_j\in I_\lambda$ are sub-exponentially decaying in the sense that
 $$
 \lim_{n\to\infty} 
 \frac{\log\left(\sqrt{{\|P_n^*\psi_{j}\|}^2\,+\,{\|P_{n+1}^*\psi_{j}\|}^2\,} \right)}{n^{2-d}}\,=\,-\,\frac{\gamma_{E_j,\lambda}}{C(2-d)}\,
 $$
 \item If additionally {\rm (A2)} is satisfied and $\lim_{n\to\infty} s_n/n= C$, i.e. $\Ab_\bfs$ has uniform $2$-dimensional growth rate, 
 then,  $H_\lambda$ has almost surely
 dense pure singular continuous spectrum in $J_\lambda$ and dense pure point spectrum in $I_\lambda \setminus J_\lambda$ where
 $$
 J_\lambda\,:=\,\left\{ E\in I_\lambda\,:\,\gamma_{E,\lambda}\,/\,C\,\leq \,1/2\,\right\}\;.
 $$
 Moreover, almost surely, the random eigenvectors $\psi_j$ for the random eigenvalues $E_j\in I_\lambda \setminus J_\lambda$ are polynomially decaying in the sense that
$$
 \lim_{n\to\infty} 
 \frac{\log\left(\sqrt{{\|P_n^*\psi_{j}\|}^2\,+\,{\|P_{n+1}^*\psi_{j}\|}^2\,} \right)}{\log(n)}\,=\,-\,\frac{\gamma_{E_j,\lambda}}{C}\,.
 $$
\end{enumerate}
\end{theorem}
\noindent Note that by analyticity of $\gamma_{E,\lambda}$ in $E\in I_\lambda$ it is clear that $J_\lambda$ is a finite union of intervals. Also note that for $C\to \infty$ the set $J_\lambda$ gets larger and converges to $I_\lambda$.
However, as long as $I_\lambda=(c'_{\lambda},\lambda v_-) \cup (\lambda v_+, c_\lambda)$ is not empty we have for $E\to c_\lambda$ or $E\to C'_\lambda$ that $|h_{E,\lambda}|\to 2$ and hence $\gamma_{E,\lambda}\to \infty$.
Therefore, $I_\lambda \setminus J_\lambda$ is not empty if $I_\lambda$ is not empty. From \eqref{eq-scaling} one obtains that for any energy $E\in (-2,2)\cap I_\lambda$ and small enough $\lambda$ we get $E\in J_\lambda$, so in case (iv) there is singular continuous spectrum for small disorder.
All the results are valid for any disorder $\lambda$ in principle, however, for large disorder, the set $I_\lambda$ may be empty, cf. Remark~\ref{rem:I_l}.

As explained above one can get this type of antitree model from a random operator on $\ZZ^d$ with radial projections around the adjacency operator.
Therefore, one has the following corollary:

\begin{coro}\label{coro:main}
 Let $\Delta_d$ be the adjacency operator on $\ZZ^d$, $\Pp_r$ the orthogonal projection on the radial subspace $\HH_r=\{\psi\in\ell^2(\ZZ^d)\,:\,\psi(x)=\psi((\|x\|_1,0,\ldots,0))\,\}$ and consider the operator
 $$H_\lambda\,=\,\tfrac1d \,\Pp_r\,\Delta_d\, \Pp_r\,+\,\lambda \Vv$$ where $\Vv$ is a random i.i.d. potential with single-site distribution $\PP_v$ satisfying {\rm (A1)}.
 Then, for $d\geq 3$ the spectrum of $H_\lambda$ is almost surely purely absolutely continuous in $I_\lambda$ and it is almost surely purely singular in $I_\lambda$ for $d\leq 2$.
 Also, $I_\lambda\subset \spec_\ess(H_\lambda)$ almost surely.
\end{coro}

Part of the statements in Theorem~\ref{th:main} is that $I_\lambda$ is always part of the essential spectrum of $H_\lambda$, almost surely. In fact, even though the 
models are not ergodic, we obtain an almost sure essential spectrum.

\begin{prop}\label{prop-spec-H}
Let $H_\lambda=\Aa_\bfs+\lambda \Vv$ be the Anderson model on $\Ab_\bfs$ and let {\rm (A1)} be satisfied. We find almost surely that
$$
\spec(H_\lambda)\,=\,[-2,2]+\lambda \supp(\nu) \qtx{if}
\sup_n\,\{s_n\}\,<\,\infty
$$
and 
$$
\spec_{\ess}(H_\lambda)\,=\,\lambda \supp(\PP_v)\,\cup\,\left\{E\not\in \lambda\supp(\PP_v)\,:\,|h_{E,\lambda}|\leq 2\,\right\}\,\supset\,I_\lambda
$$
if $\liminf_{n\to\infty} s_n/n^\alpha\,>\,0$
for some $\alpha>0$, i.e. if $\Ab_\bfs$ has at least $1+\alpha$ dimensional growth rate.
In the latter case one may find random eigenvalues in $[-2,2]+\lambda \supp(\PP_v) \setminus \spec_{\ess}(H_\lambda)$.
\end{prop}

\begin{remark} \label{rem:I_l}
In general, note that under assumption {\rm (A1)} for $E\not\in[ \lambda v_-\,,\, \lambda v_+]$ and $v\in\supp(\PP_v)\subset[v_-,v_+]$ the sign of $E-\lambda v \in [E-\lambda  v_+,E-\lambda  v_-]$ is determined and 
it is clear that $h_{E,\lambda}\in [E-\lambda  v_+,\,E-\lambda  v_-]$. Moreover, by convexity $|h_{E,\lambda}|\leq|E|$ and
$h_{E,\lambda}$ is strictly monotonically increasing in $E$. It is not difficult to verify $dh_{E,\lambda}/dE>1$. 
Thus, $I_\lambda$ generally consists of two intervals
$$
[-2,\lambda v_-)\cup(\lambda v_+,2]\subset
I_\lambda\,=\,(c'_{\lambda},\lambda v_-) \cup (\lambda v_+, c_\lambda)\,\subset\,(-2+\lambda v_-,\lambda v_-)\cup (\lambda v_+, 2+\lambda v_+)\,,
$$
where $h_{c'_\lambda,\lambda}=-2$ and $h_{c_\lambda,\lambda}=2$.
For large $\lambda$ these values $c'_\lambda<\lambda v_-$ and $c_\lambda > \lambda v_+$ may or may not exist in the sense that each of the intervals $(c'_{\lambda},\lambda v_-)$ and $(\lambda v_+, c_\lambda)$ may be empty. More precisely, if $\EE_v(v_+-v)^{-1}=\infty$  
then $\lim_{E\downarrow \lambda v_+} h_{E,\lambda}=0$ and for any $\lambda>0$, $(\lambda v_+,  c_\lambda)\neq\emptyset$. 
However, if $\EE_v(v_+-v)^{-1} <\infty$, then
 $I_\lambda\cap \RR_+$ is empty for large $\lambda$. Correspondingly, $I_\lambda\cap \RR_-$ is empty for large $\lambda$ if
 $\EE_v(v-v_-)^{-1} <\infty$ and it is never empty if $\EE_v(v-v_-)^{-1}=\infty$.
If $\supp(\PP_v)=[v_-,v_+]$, then the essential spectrum has no gap. However, if there is a gap in $\supp(\PP_v)$, then for any value $e$ in that gap and large enough $\lambda$ 
one will have $|h_{e\lambda,\lambda}|=|\lambda h_{e,1}|>2$. Hence, a gap in $\spec_\ess(H_\lambda)$ will open at some $\lambda$.
The intersection of these gaps with $[-2,2]+\lambda\supp(\PP_v)$ contains random eigenvalues and can be seen as a Lifshitz tail regime.
 For illustration, let us consider the following cases where $v_+=1$ and $v_-=-1$.
 \begin{enumerate}[{\rm (i)}]
  \item For the Bernoulli distribution $\PP_v=\frac12(\delta_{-1}+\delta_1)$ which satisfies {\rm (A1)} but not {\rm (A2)} ($\delta_x$ denotes the normalized point measure at $x$), we find
  \begin{align*}
  &I_\lambda\,=\,\left(-1-\sqrt{1+\lambda^2}\,,\,-\lambda \right)\,\cup\,\left(\lambda\,,\,1+\sqrt{1+\lambda^2} \right)\,,\quad
  h_{E,\lambda}\,=\,\frac{E^2-\lambda^2}{E}\\
  &\spec_{\ess}(H_\lambda)\,=\,\left[-1-\sqrt{1+\lambda^2}\,,\,1-\sqrt{1+\lambda^2} \right]\,\cup\,\left[1-\sqrt{1+\lambda^2}\,,\,1+\sqrt{1+\lambda^2} \right]
  \end{align*}
  \item For the uniform distribution $\PP_v(dv)=\frac12\, 1_{[-1,1]}(v)\,dv $, where $1_A(v)$ denotes the indicator function of the set $A$ and 
  $dv$ the Lebesgue measure, we find
 \begin{align*}
  I_\lambda\,=\,\left(-\lambda \frac{e^\lambda+1}{e^\lambda-1}\,,\,-\lambda\right)\,\cup\,\left(\lambda\,,\,\lambda \frac{e^\lambda+1}{e^\lambda-1}\,\right)\;,\qquad
  &h_{E,\lambda}\,=\,\frac{2\lambda}{\ln(\frac{E+\lambda}{E-\lambda})}\\
  \spec_{\ess}(H_\lambda)\,=\,\left[-\lambda \frac{e^\lambda+1}{e^\lambda-1}\,,\,\lambda \frac{e^\lambda+1}{e^\lambda-1}\,\right]\,.
  \end{align*}
  \item For $\PP_v(dv)= 1_{[-1,1]}(v)\,(1-|v|)\,dv$ and $E\not\in[-\lambda,\lambda]$ we find with $x:=\lambda/E \in (-1,1)$ that
  $$
  h_{E,\lambda}\,=\,\frac{\lambda\; x}{ (1+x) \ln\left(1+x\right)\,+\,\left(1-x \right) \ln\left(1-x\right) }\,.
  $$
  Therefore, $|h_{E,\lambda}|> \lambda / (2\ln(2))$ for $|E|>\lambda$, and $I_\lambda$ is empty for $\lambda\geq 4 \ln(2)$. For $0\leq \lambda < 4\ln(2)$ we find $I_\lambda=(-c_\lambda,-\lambda)\cup(\lambda,c_\lambda)$ and $\spec_\ess(H_\lambda)=[-c_\lambda,c_\lambda]$ and for $\lambda\geq 4\ln(2)$ we have $\spec_\ess(H_\lambda)=[-\lambda,\lambda]$.
  
  \item For general distributions satisfying {\rm (A1)} we find for $0<\lambda\leq 2+\frac{2\sigma^2}{2-\sigma^2}$ that
  \begin{align*}
  & \left\{E\in\RR_\pm\,:\,\lambda v_\pm < |E| \leq 2+\tfrac{\lambda^2\sigma^2}{10}\right\}\,\subset\,\\
  & \left\{E\in\RR_\pm\,:\,\lambda v_\pm < |E| \leq 1-\tfrac{\lambda}{2}+\sqrt{(1+\tfrac{\lambda}{2})^2+\lambda^2 \sigma^2}\right\} \,\subset\, I_\lambda\,.
  \end{align*}
 \end{enumerate}
To obtain part {\rm (iv)} let $X=E-\lambda v$, $h=h_{E,\lambda}$ $\EE=\EE_v$ and note that
$h=\EE(h^2/X)$. Hence, 
$$
|E|-|h|=|E-h|=|\EE(X-h)|=\left|\EE\left(\frac{(X-h)^2}{X}\right)\right| \,\geq\, \frac{\EE((X-h)^2)}{|E|+\lambda}\,=\,
\frac{\lambda^2 \sigma^2+(\EE(X-h))^2}{|E|+\lambda}\,.
$$
Thus, $|h_{E,\lambda}| < |E|-\frac{\lambda^2\sigma^2}{|E|+\lambda}$. 
Analyzing $|E|-\frac{\lambda^2\sigma^2}{|E|+\lambda}\leq 2$ one finds the second relation in {\rm (iv)}.
For the first relation note that $\sqrt{a^2+b^2}\geq a+ \tfrac{1}{2\sqrt{a^2+b^2}} b^2$ and 
$\sqrt{(1+\lambda/2)^2+\sigma^2\lambda^2}\leq 5$ for $\lambda\leq 2+\frac{2\sigma^2}{2-\sigma^2}$ and $\sigma^2\leq 1$.
Note that part (iv) shows the statement on $I_\lambda$ being non-empty in Theorem~\ref{th:main}.
\end{remark}

Let me give some overview of the proofs.
The operators considered in Theorem~\ref{th:main} have a special structure which allows an analysis through $2\times 2$ transfer matrices. 
We say that such operators have one propagating channel, for a precise definition see Section~\ref{sub:spectrum}.
In fact, one of the novelties in this work is the realization that many techniques from one-dimensional theory such as
subordinacy theory by Gilbert-Pearson and Kahn-Pearson \cite{GP, KP} and links between transfer matrices and spectrum as developed by
Kiselev, Last and Simon \cite{LaSi, KLS} can be translated to operators with one propagating channel.
The precise theorems are listed in Section~\ref{sub:spectrum} and details are carefully carried out in Appendix~\ref{app:spectrum} which shows the modifications
that have to be made compared to the Jacobi operator case. For instance, not for all real energies all the transfer matrices will be defined. The set of energies where some transfer matrix is not defined will be denoted by $B_\infty$ and in order to make a direct translation of the subordinacy theory, one needs to stay away from the closure of $B_\infty$. Moreover, some new technical estimates are needed which are completely trivial in the Jacobi case
(cf. Lemma~\ref{lem-sub-keyest} and the proof in Section~\ref{sec:key-estimate-sub}).
The kinetic part of $H_\lambda$ effectively averages the potentials in the $n$-th shell within the energy regions $I_\lambda$ which is expressed in terms of a harmonic mean of random variables in the transfer matrices (cf. \eqref{eq-harm-entry}).
Then one needs some estimates as done in Section~\ref{sec:keyestimate} and the
spectral analysis tools translated from the one-dimensional theory, to 
see that one can treat the problem analogue to random decaying potentials in one dimension as done in \cite{KLS}.
One important step is to show that it is in principle sufficient to consider the spectral measure at the roots $S_0$, cf. Theorem~\ref{th:sigma_0}.
Here, this is not completely trivial because unlike in the pure one dimensional case $\Ab_\bfs=\ZZ_+$ ($s_n=1$ for all $n$), the localized states at the roots are not necessarily cyclic for $H_\lambda$, cf. Remark~\ref{rem:purepoint}~(ii). Starting from identities between the Green's functions and formal solutions of the
eigenvalue equation (cf. Lemma~\ref{lem:gr}) one can still show that for the energies where all transfer matrices are defined it is sufficient to consider
the measure at a special vector supported on $S_0$.

\vspace{.2cm}

The paper is organized as follows. Section~\ref{sub:spectrum} lists all adapted theorems from the theory of one-dimensional Jacobi operators to the more
general setup of operators with one propagating channel, they are deterministic (no randomness enters) and may be interesting on their own.
The translation of subordinacy theory requires some new estimates (Lemma~\ref{lem-sub-keyest}) proved in Section~\ref{sec:key-estimate-sub}.
Details of the adaption of these theorems are carefully carried out in Appendix~\ref{app:spectrum}.

In Section~\ref{sec:spec-H} we prove Proposition~\ref{prop-spec-H} based on Proposition~\ref{prop-spec_ess} also proved there. 
Section~\ref{sec:keyestimate} summarizes some facts on harmonic means of random variables and in Section~\ref{sec:pruefer} we conclude proving Theorem~\ref{th:main}. There we combine
some of the theorems of Section~\ref{sub:spectrum} and techniques as in \cite[Section~8]{KLS}, particularly Theorem~\ref{th-lim-log(T)} where some more details are given in Appendix~\ref{app-estimate}.
Finally, in Section~\ref{sec:PDP} we prove Corollary~\ref{coro:main}.

\subsection{Some open questions}


In the set $\Ii_\lambda:=[\lambda v_-,\lambda v_+]$ the random variables $E-\lambda v(x)$ do not have a distinct sign. The entries in the transfer matrices  (cf. \eqref{eq-def-Tt}, \eqref{eq-def-psi-T}) will be harmonic means of such random variables (cf. \eqref{eq-harm-entry}). Depending on the energy and the distribution of $\PP_v$, these harmonic means may not have an expectation and also $h_{E,\lambda}$ does not need to exist. 
Therefore, the transfer matrices have a lot of randomness and one might expect localization in this region.
On the other hand, as $s_n\to\infty$ the set $B_\infty$ (cf. \eqref{eq-def-A-B}) where not all transfer matrices are defined is dense in 
the interior of $\lambda\supp(\PP_v)$. 
This sort of also reflects the fact that we have an eigenvalue at $0$ with infinite multiplicity for $\lambda=0$. For positive $\lambda$ these 'states' may start to resonate and one could imagine the formation of some delocalized states.
However, similar as in a recent paper by Aizenman, Shamis and Warzel on resonances and partly delocalized states on the complete graph {\rm \cite{AShW}}, these states may only delocalize in the $\ell^1$ but not $\ell^2$ sense meaning that one could have eigenvectors not lying in $\ell^1(\Ab_\bfs)$. 
A deeper analysis of these resonances might be very interesting. The spectral type may also depend on the distribution $\PP_v$ or be random,
particularly in the region $\Ii_\lambda\setminus \lambda\supp(\PP_v)$.

 Even though the set $I_\lambda$ may not be empty even for large $\lambda$ and so one may have some a.c. spectrum if $\sum_n s_n^{-1}<\infty$,
there still should be some sort of large disorder localization:
 \begin{conjecture}
 For any compact interval $[a,b]$ and any sequence $\bfs$, there is a $\lambda_0>0$ such that
 for $\lambda>\lambda_0$ the spectrum of the Anderson model $H_\lambda$ on $\Ab_\bfs$ is almost surely pure point in $[a,b]$.
 \end{conjecture}
 This conjecture is intertwined with the comment above as for any compact set and large enough $\lambda$ one has $[a,b]\subset \Ii_\lambda$.
This conjecture is trivially true if $0\not\in\supp(\PP_v)$ because of a gap in the spectrum for large $\lambda$.
In terms of the above theorems one may also conjecture the following:
  \begin{conjecture}
 If  $\sum_n s_n^{-1}=\infty$, then the spectrum of $H_\lambda=\Aa_\bfs+\lambda \Vv$ is almost surely singular at any disorder.  
 \end{conjecture}


\section{Operators with one propagating channel \label{sub:spectrum}}

Now let us examine the structure of the operator $\Aa_\bfs+\lambda\Vv$ in more detail.
We will use the equivalence $\ell^2(\Ab_\bfs)=\bigoplus_n \ell^2(S_n)\cong \bigoplus_n \CC^{s_n}$ and can therefore write
$\psi=\bigoplus_{n=0}^\infty \psi(n)$ where $\psi(n)=(\psi(n,1),\ldots,\psi(n,s_n))^\top \in \CC^{s_n}$. The transpose shall emphasize that we will consider $\psi(n)$ as a column vector and the pairs $(n,j)$, $n\in\ZZ_+, j=1,\ldots,s_n$ denote the vertices in $\Ab_\bfs$.

Define $D_n\in\Mat(s_n\times s_{n-1})$ by $(D_n)_{jk}=\langle \delta_{n,j}\,|\,\Aa_\bfs\,\delta_{n-1,k}\rangle$
where $\delta_{n,j}$ is the normalized $\ell^2(\Ab_\bfs)$ vector with entry one at $(n,j)$ and entry zero on all other vertices and 
$\langle\cdot |\cdot\rangle$ denotes the scalar product.
Then, as each vertex in $S_n$ is only connected to vertices in $S_{n\pm 1}$, we find $(\Aa_\bfs \psi)(n)=D_{n+1}^*\psi(n+1)+D_n \psi(n-1)$.
 Note that
\begin{equation}\label{eq-def-phi}
D_n\,=\,\frac{1}{\sqrt{s_n s_{n-1}}} \pmat{1 & \cdots & 1\\ \vdots & & \vdots \\ 1 & \cdots & 1}\,=\,
\phi_n \phi_{n-1}^*\qtx{where} \phi_n=\frac{1}{\sqrt{s_n}} \;\pmat{1\\\vdots\\1}\,\in\, \CC^{s_n}\,,
\end{equation}
and the term $\phi_n\phi_{n-1}^*$ has to be understood as a matrix product of a colum vector with a row vector.
Therefore, the Anderson model can be written as
\begin{equation}\label{eq-H_lb}
(H_\lambda \psi)(n)\,=\,\phi_n (\phi_{n+1}^*\psi(n+1)+\phi_{n-1}^*\psi(n-1))+\lambda\,V_n \psi(n)
\end{equation}
where $V_n=\diag(v_{n,1},\ldots,v_{n,s_n})$
and the $v_{n,j}$ are real, independent identically distributed random variables with distribution $\PP_v$. 
Note, as the $\phi_n$ are vectors, the expression $\phi_n^* \psi(n)$ is the standard scalar product between $\phi_n$ and $\psi(n)$
in $\CC^{s_n}$. When necessary we will consider the $v_{n,j}$ as random variables on an abstract probability space $(\Omega,\Aa,\PP)$.
Expectations with respect to $\PP$ will be denoted by $\EE$, so $\EE(v_{n,j})=0,\,\EE(v_{n,j}^2)=\sigma^2$.
At $0$ we have Dirichlet boundary conditions, i.e. formally $s_{-1}=1,\,\psi(-1)=0$ and $\phi_{-1}$ is any number.

We will use some generalizations of well-known one-dimensional transfer matrix techniques and list the
corresponding theorems here. For the proofs one has to go through the one-dimensional theory and adapt the proofs and results step by step.
As the situation here is different 
from the pure one-dimensional case we give the details in Appendix~\ref{app:spectrum}.

Because of the more general nature of these theorems we drop the coupling constant $\lambda$ and let $V_n\in \Her(s_n)$ be general Hermitian matrices and $\phi_n\in\CC^{s_n}$
be any sequence of {\bf non-zero} vectors such that the operator
\begin{equation} \label{eq-H-prop}
(H\psi )(n)\,=\, \phi_n\left( \phi_{n+1}^* \psi(n+1)\,+\,\phi_{n-1}^* \psi(n-1) \right)\,+\, V_n \psi(n)\;
\end{equation}
acting on $\psi=\bigoplus_{n\geq 0} \psi(n)\,\in\,\bigoplus_{n\geq 0} \CC^{s_n} \cong \bigoplus_{n\geq 0} \ell^2(S_n)=\ell^2(\Ab_\bfs)$
is uniquely self-adjoint and 
\begin{equation}\label{eq-def-D0}
\Dd_0=\Big\{\psi=\bigoplus_{n\geq 0}\psi(n) \in \bigoplus_{n\geq 0} \CC^{s_n}\,:\, \psi(n)=0\;\text{for all but finitely many}\,n\Big\}
\end{equation}
is a core. By Remark~\ref{rem:selfad} this condition is satisfied if $\sum_{n=0}^\infty \|\phi_n\phi_{n+1}^*\|^{-1}=\infty\;.$

Now for $n\in\ZZ_+$ let $\Phi_n=\bigoplus_k\Phi_n(k) \in \ell^2(\Ab_\bfs)$ be defined by
\begin{equation}\label{eq-def-Phi}
\Phi_n=P_n\phi_n \qtx{i.e.} \Phi_n(n)=\phi_n\qtx{and} \Phi_n(k)=\nul\quad\text{for\;} n\neq k\;,
\end{equation}
where as above $P_n$ is the canonical injection of $\CC^{s_n}\cong\ell^2(S_n)$ into $\bigoplus_k\CC^{s_k} \cong \ell^2(\Ab_\bfs)$.
In the direct sum notation used above this means $\Phi_n = \bigoplus_{k=0}^{n-1} \nul \oplus \phi_n\oplus \bigoplus_{k>n} \nul$.

\begin{defini}
 A self-adjoint operator of the form \eqref{eq-H-prop} is said to have one propagating channel defined by the sequence of the vectors $\phi_n$.
\end{defini}
Let us explain this notion. Consider the kinetic part $H_0$ (all $V_n$ equal zero), then the modes of the $n$-th shell connected to $n-1$ shell 
are given by $\ran(P_n^* H_0 P_{n-1})$ and the modes of the $n$-th shell connecting the $n$-th with the $n+1$-st shell are given by $\ran(P_n^* H_0 P_{n+1})$.
In this case, both are given by the one-dimensional space spanned by $\phi_n$, i.e. these modes propagate through the $n$-th shell. 
The sequence $(\phi_n)_n$ forms a channel through which quantum waves can travel.
An obvious generalization is to have an operator as in \eqref{eq-H-prop} where the $\phi_n$ are $s_n \times k$ matrices of rank $k$. 
Then $\ran(P_n^* H_0 P_{n-1})=\ran(P_n^* H_0 P_{n+1})$ would be always a $k$-dimensional space and one would have $k$ propagating channels.

Let us start the analysis with the following trivial facts. 
\begin{prop}\label{prop-spec-decom}
Let $\VV_n\subset\CC^{s_n} \cong \ell^2(S_n)$ be the cyclic space of $\phi_n$ w.r.t. to $V_n$, i.e. $\VV_n$ is the span of
$V_n^k\phi_n$ for $k=0,\ldots,s_n-1$. Then one has the following:
\begin{enumerate}[{\rm (i)}]
 \item The direct $\ell^2$ sum $\VV:=\bigoplus_{n\geq 0} \VV_n\subset\ell^2(\Ab_\bfs)$ equals the cyclic space generated by the family $(\Phi_n)_{n\geq 0}$.
One has $\VV=\ell^2(\Ab_\bfs)$ if and only if $\phi_n$ is a cyclic vector for $V_n$ for all $n\in\NN$.
 \item $H$ leaves the space $\VV$ and the orthogonal components in $\ell^2(S_n)$, i.e. $\VV^\perp\cap \ell^2(S_n)\cong\VV_n^\perp$ invariant. In particular,
 the spectrum $\spec(H)=\spec(H|\VV)\cup \bigcup_n \spec(H|\VV_n^\perp)$ where $\spec(H|\VV_n^\perp)$ consists of finitely many eigenvalues with corresponding eigenvectors in $\ell^2(S_n)$.
\end{enumerate}
\end{prop}
If $V_n=\nul$ and $\|\phi_n\|=1$ for all $n$, then we find $\VV_n=\phi_n\CC$ and $H|\VV$ is isomorphic to the adjacency operator on
$\ZZ_+$ with pure a.c. spectrum on $[-2,2]$.  This shows Proposition~\ref{prop-H0}.

The eigenvalue equation $H\psi=z\psi$ with $\psi(n)\in\VV_n$ can be written as
\begin{equation}\label{eq-eig-psi}
(z-V_n)\,\psi(n)\,=\, \phi_n \left( \phi_{n+1}^* \psi(n+1) + \phi_{n-1}^* \psi(n-1)\right)\,.
\end{equation}
Let us define
$u_n=u_n(\psi):=\phi_n^* \psi(n)\,\in\, \CC$ and let $z \not\in \spec(V_n|\VV_n)$ which is always the case for non-real energies.
Then $(z-V_n)^{-1} \phi_n \in \VV_n$ is well defined, even if $z\in\spec(V_n|\VV_n^\perp)$. 
For a solution $\psi$ of \eqref{eq-eig-psi} with $z\not\in\spec(V_n|\VV_n)$,  $n\geq1$, we obtain for $u_n=u_n(\psi)$ that
\begin{equation}\label{eq-eig-u}
 u_n\,=\,(\phi_n^* (z-V_n)^{-1}\phi_n)\,(u_{n+1}+u_{n-1})\;.
\end{equation}
Assuming $\phi_n^*(z-V_n)^{-1} \phi_n$ is not zero, this can be rewritten as
\begin{equation} \label{eq-def-Tt}
 \pmat{u_{n+1} \\ u_n}\,=\,T_{z,n}\pmat{u_n \\ u_{n-1}}\qtx{where}
 T_{z,n}:=
 \pmat{\left(\phi_n^* (z-V_n)^{-1} \phi_n\right)^{-1} & -1 \\ 1 & 0 }\;.
\end{equation}
The $T_{z,n}$ are the $2\times 2$ transfer matrices at stage $n$ associated to the operator $H$.
By this equation for $n=0$ we may also define $u_{-1}$ for such a solution.
Again, for complex $z\not\in\RR$, $\phi_n^* (z-V_n)^{-1} \phi_n$ exists and is invertible as the imaginary part will be negative.
For $\alpha\in\VV_n, V_n\alpha=E\alpha$, multiplying \eqref{eq-eig-psi} with $\alpha^*$ from the left gives
$ 0\,=\,\alpha^*\phi_n (u_{n+1}+u_{n-1})$. Therefore, $u_{n+1}=-u_{n-1}$ for a  solution if $E\in\spec(V_n|\VV_n)$ (then $\alpha\in\VV_n$, so $\alpha^*\phi_n\neq 0$). Hence, define
\begin{equation}\label{eq-def-T_E}
 T_{E,n}=\pmat{0&-1\\1&0}\qquad\text{if $E\,\in\,\spec(V_n|\VV_n)$.}
\end{equation}
In fact, \eqref{eq-def-T_E} extends $z\mapsto T_{z,n}$ holomorphically to $z\in\spec(V_n|\VV_n)$. 
Furthermore, let us define
\begin{equation}\label{eq-def-A-B}
 A_n\,:=\,\{E\in\RR\,: \,\phi_n^*(E-V_n)^{-1} \phi_n=0\}\,,\quad
 B_n\,:=\,\bigcup_{k=0}^{n-1} A_k\,,\quad
 B_\infty\,:=\,\bigcup_{k=0}^\infty A_k\;.
\end{equation}
So $A_n$ is exactly the finite set where $T_{E,n}$ is not defined and consists of $\dim(\VV_n)-1$ points interlaced between the eigenvalues\footnote{Note that by cyclicity, the eigenvalues of $V_n|\VV_n$ are simple.} of $V_n|\VV_n$.

For any complex $z\not\in B_n$ we can define the transfer matrix $T_z(n)$ from $0$ to $n$ by
\begin{equation}
 T_z(n)\,:=\, T_{z,n-1} \,\cdots\, T_{z,1}T_{z,0}\;,\qtx{then}  
 \pmat{u_{n} \\ u_{n-1}}\,=\,T_{z}(n)\pmat{u_0 \\ u_{-1}}\;.
\end{equation}
For convenience we also define
\begin{equation}\label{eq-def-a-psi}
a_{z,n}:= \frac{1}{\phi_n^*(z-V_n)^{-1} \phi_n}\;,\qquad 
\psi_{z,n}:= \frac{(z-V_n)^{-1}\phi_n}{\phi_n^*(z-V_n)^{-1} \phi_n}.
\end{equation}
Note that
\begin{equation}\label{eq-rel-abpsi}
 \im(a_{z,n})\,/\,\im(z)\,=\,\|\psi_{z,n}\|^2\,\geq\,\left|\frac{d}{dz} a_{z,n} \right| \,\geq\,\frac{1}{\|\phi_n\|^2}\qtx{and} 
 \|\psi_{E,n}\|^2\,=\,\frac{d}{dE} a_{E,n}\;.
\end{equation}
As above, if $z=E\in \spec(V_n|\VV_n)$ then we let $a_{E,n}=0$ and $\psi_{E,n}$ can be defined by analytic extension
leading to  $\psi_{E,n}=\alpha / (\phi_n^* \alpha)$ where $\alpha$ is a unit eigenvector for the eigenvalue $E$ of $V_n|\VV_n$
(cf. proof of Lemma~\ref{lem:msr}~(ii)). 
Hence, $a_{z,n}$ and $\psi_{z,n}$ are defined for $z\not\in A_n$.

First let us consider the spectrum as a set. Following the transfer matrices let us define the corresponding Jacobi operator $\Jj_z$ on $\ell^2(\ZZ_+)$ 
for any $z\not\in B_\infty$ by
\begin{equation}\label{eq-def-Jj_z}
 (\Jj_z u)_n\,=\,u_{n+1}\,+\,u_{n-1}\,-\,a_{z,n}\,u_n.
\end{equation}
Then $\Jj_z$ has the same transfer matrices at energy $0$ and $\Jj_E$ is self-adjoint for real energies $E\not\in B_\infty$.
\begin{prop}\label{prop-spec_ess}
  Assume $\sup_n \|\phi_n\|<\infty$ and $E\not\in B_\infty$. Then, $E\in\spec(\Jj_E) \,\Rightarrow\,E\in\spec(H)$ and 
  $E\in\spec_\ess(\Jj_E)\,\Rightarrow\,E\in\spec_\ess(H)$. If moreover $E\not\in \spec(\Vv)$ where $\Vv=\bigoplus_{n=0}^\infty V_n$, then we have equivalence:
  $E\in \spec(H) \;\Leftrightarrow\; 0\in\spec(\Jj_E)$ and
$E\in \spec_{\ess}(H) \;\Leftrightarrow\; 0\in\spec_{\ess}(\Jj_E)$.
\end{prop}

Let $\mu_n$ denote the spectral measure at $\Phi_n$, i.e.
\begin{equation}
 \int \,f\, d\mu_n\,=\,\langle\Phi_n\,|\,f(H)\,|\,\Phi_n\rangle\;.
\end{equation}
Following the relations between Green's functions and solutions to \eqref{eq-eig-u} and using arguments by Carmona \cite{Car, CL} 
we find the following relations of spectral measures. Part (ii) and (iii) are analogue to the one-dimensional theory.
\begin{theorem}\label{th:sigma_0}
 We have the following:
\begin{enumerate}[{\rm (i)}]
\item Let $\varphi \in \VV_n\subset \CC^{s_n}$ and define $\boldsymbol{\varphi}=\bigoplus_k \bphi(k) \in \bigoplus_k \CC^{s_k}$ by $\bphi(n)=\varphi$, $\bphi(k)=\nul\in \CC^{s_k}$ for $k\neq n$.
Moreover, let $\mu_{n,\varphi}$ denote the spectral measure at $\boldsymbol{\varphi}$, i.e. $\int f d\mu_{n,\varphi}=\langle\boldsymbol{\varphi}|f(H)|\boldsymbol{\varphi}\rangle$.
Then, on $\RR \setminus A_n$ the measure $\mu_{n,\varphi}$ is absolutely continuous with respect to $\mu_n$ and one has
$$
1_{\RR\setminus A_n}(E) \mu_{n,\varphi}(dE)\,=\,1_{\RR\setminus A_n}(E) \;|\varphi^* \psi_{E,n}|^2\,\mu_n(dE)
$$

Moreover, the set of energies where $\varphi^*\psi_{E,n}=0$ is finite.
\item
On the set where $T_E(n)$ is well defined, $\RR\setminus B_n$, 
the measure $\mu_n$ is absolutely continuous with respect to $\mu_0$
and we have
$$
1_{\RR\setminus B_n}(E)\;\mu_n(dE)=1_{\RR \setminus B_n}(E)\;\big|\smat{1&0} T_E(n)\smat{1\\0}\big|^2\;\mu_0(dE)\,.
$$
In particular, $\mu_0$ is a spectral measure for $H|\VV$ on $\RR\setminus B_\infty=\RR\setminus\left(\bigcup_{k=0}^\infty A_k\right)$.
\item There exists a positive point measure $\nu$ supported on $B_\infty$ such that $\mu_0$ is given by the weak limit
$$
\mu_0(dE) \,=\, \lim_{n\to\infty} \frac{1_{\RR\setminus B_\infty}(E)\,dE}{\pi\,\|T_E(n) \smat{1\\0}\|^2}\;+\:\nu(dE).
$$
The measure $\nu$ includes a delta measure at $E\in B_\infty$ if and only if for the smallest integer $m$ such that $E \in A_{m}$ one finds that $T_E(m)\smat{1\\0}=\smat{0\\ c}$ for some $c$ (cf. Remark~\ref{rem:purepoint}~{\rm (iii)} ).
\end{enumerate}
\end{theorem}

\begin{remark}\label{rem:purepoint} $\phantom{h}$
\begin{enumerate}[{\rm (i)}]
\item Note as $A_n$ is finite, $\mu_n$ is pure point whenever $\mu_{n,\varphi}$ is pure point for some $\varphi\in\VV_n$.
Similarly, as $B_\infty$ is countable and $\VV_n^\perp$ finite dimensional one immediately sees that $H$ has pure point spectrum, whenever the measure
 $\mu_0$ is pure point. 

\item One might get the impression that $\Phi_0$ should be a cyclic vector for $\VV$, however, this does not need to be the case.
 It is possible to have an eigenvalue in $B_\infty$ with an eigenvector $\psi \in \VV$ which is orthogonal to $\Phi_0$.
 To see this assume that $E\in A_m$, i.e. $\Phi_m^* (V_m - E)^{-1} \Phi_m =0$ but $E\not \in A_n$ for all $n>m$, thus $T_{E,n}$ exists.
 Let $u_m=0,\,u_{m+1}=1$ and $\smat{u_{n+1}\\u_n}=T_{E,n} \smat{u_n\\u_{n-1}}$ for $n>m$ and assume that $\sum_{n>m} |u_n|^2 \|\psi_{E,n}\|^2<\infty$.
Then let $\psi(n)=u_n \psi_{E,n}$ for $n>m$, $\psi(m)=(E-V_m)^{-1}\phi_m u_{m+1}$ and $\psi(n)=0$ for $n<m$.
 It is easy to check that $H\psi=E\psi$ and $\psi$ is orthogonal to $\Phi_n$ for $n=1,\ldots,m$.

 \item In part {\rm (iii)} of the above theorem one can construct an eigenvector contributing to $\nu$ in a similar fashion. For $T_E(m)\smat{1\\0}=\smat{0\\c}$, $E\in A_m\subset B_\infty$, 
 let $u_{-1}=0$,\,$u_1=1$ and
$\smat{u_{n}\\u_{n-1}}=T_E(n)\smat{1\\0}$. Then set $\psi(m)=(E-V_{m})^{-1} \phi_{m}\,u_{m-1}$,
$\psi(n)=0$ for $n>m$,  $\psi(n)=u_n \psi_{E,n}$ for $n<m$. Using all the assumptions one easily verifies
$H\psi=E\psi$ and $\langle\Phi_0|\psi\rangle=u_1=1$. 
Moreover, any eigenvector must have $\phi_{m}^* \psi(m)=0$ as can be seen from \eqref{eq-eig-psi}, thus the condition $T_E(m)\smat{1\\0}=\smat{0\\ c}$ is really needed. This eigenvector is an eigenvector for any cutoff of $H$ at $N>m$ and any boundary condition and hence contributes to $\nu$ (cf. proof in Section~\ref{sec:sigma_0}).

\item Constructing eigenfunctions similarly as in {\rm (ii)} and {\rm (iii)} one sees that $E\in B_\infty$ can be a multiple eigenvalue of $H|\VV$.
For $E\in A_n\cap A_m,\; n<m$ and $E\not\in A_k$ for $n<k<m$ one may construct an eigenvector $\psi_{n,m}$ supported from the $n$-th to the $m$-th shell
iff a solution of \eqref{eq-eig-u} staisfies both boundary conditions $u_n=0=u_m$. Allow $n=-1$ to denote eigenvectors as in {\rm (iii)}.
As $\psi_{-1,n}$ and $\psi_{n,m}$ are not orthogonal,
more than one such eigenvector can contribute to $\nu(\{E\})$.
\end{enumerate}
\end{remark}

To ensure pure a.c. spectrum we will use the following theorem, a version of \cite[Theorem~1.3]{LaSi}. 
It follows directly from Theorem~\ref{th:sigma_0} and \cite[Lemma 3.8]{LaSi} with essentially the same proof as
\cite[Theorem~1.3]{LaSi}
\begin{theorem}\label{th:sigma_0-ac}
Assume that the transfer matrices $T_{E,n}$ exist for all $E\in[a,b]$ and all $n\in\NN$, i.e. $[a,b]\cap B_\infty=\emptyset$, and assume for some $p>2$ one has
 $$
 \liminf_{n\to\infty} \int_{a}^b \|T_{E}(n)\|^p\,d_E\,<\,\infty\;.
 $$
Then, the spectrum of $H$ on the cyclic space $\VV$ generated by the $\{\Phi_n\,:\,n\in\NN\}$ is purely absolutely continuous in $(a,b)$.
\end{theorem}

\vspace{.2cm}

The subordinacy theory of Gilbert-Pearson \cite{GP}, or better Kahn-Pearson \cite{KP} also translates to some extend. 
However, there are some differences. First, the actual solution of the eigenvalue equation is $\bigoplus_n u_n \psi_{E,n}$, so we need to adjust the norm of the sequence accordingly and define:
\begin{defini} \label{def:subordinate}
Let $E\not \in B_\infty$.
 A solution $w=(w_n)_n$ to the modified eigenvalue equation \eqref{eq-eig-u} at energy $E$ will be called subordinate iff for all linear independent solutions $v=(v_n)_n$ one has
 $$
 \lim_{n\to\infty} \,\frac{\|w\|_{E,n}}{\|v\|_{E,n}}\,=\,0 \qtx{where}
 \|w\|^2_{E,n}\,:=\,\sum_{k=0}^n |w_k|^2\, \|\psi_{E,k}\|^2\;.
 $$
\end{defini}

But also after this adjustment we require some more estimates for the subordinacy theory.
This is somehow related to the fact that $\Phi_0$ may not be cyclic for the space $\VV$, the cyclic space generated by all the $\Phi_n$. 
Energies in $B_\infty$ may lead to eigenvectors orthogonal to $\Phi_0$.
Following the proofs in \cite{KP} closely, it turns out that we need the following estimates
which may also be useful in other circumstances. As this estimate is a new ingredient (which is trivial in the pure 1D case), we singled out its proof in Section~\ref{sec:key-estimate-sub}.
\begin{lemma}\label{lem-sub-keyest} We find the following estimates:
\begin{enumerate}[{\rm (i)}]
\item For any $E\not \in B_\infty$ we have
\begin{equation}\label{eq-imp-est}
 \frac{|a_{E+i\eta,n}-a_{E,n}|}{\|\psi_{E,n}\|^2} 
 \,=\, \frac{|a_{E+i\eta,n}-a_{E,n}|}{\left|\,(\frac{d}{dz} a_{z,n})|_{z=E}\right|} 
 \,\leq\,\eta\;,\qquad \|\psi_{z,n}\|\,\leq\,\|\psi_{E,n}\|\;.
\end{equation} 
\item Assume that $(E-3\varepsilon\,,\,E+3\varepsilon) \cap B_\infty = \emptyset$ for some $\varepsilon>0$.
Then we find a uniform constant $C>0$ such that for all $0\leq \eta \leq \varepsilon$ and all $n\in \ZZ_+$ one has
\begin{equation}
1\,\leq\, \|\psi_{E,n}\|\,/\,\|\psi_{E+i\eta,n}\|\,\leq\,C\;.
\end{equation}
\end{enumerate}
\end{lemma}
\vspace{.2cm}

The second estimate is only available for energies not in the closure $\overline{B_\infty}$ of $B_\infty$.
Indeed, the ratio $\|\psi_{E,n}\|/\|\psi_{E+i\eta,n}\|$ may blow up for any $\eta>0$ along subsequences in $n$
where $E$ gets arbitrarily close to the sets $A_n$ (in which case $E\in\overline{B_\infty}$).
So we need to restrict the subordinacy characterization to the complement of $\overline{B_\infty}$.
Together with these estimates one can follow the paper by Kahn and Pearson \cite{KP} as explained in Appendix~\ref{sub:subordinacy}.
To state the result, let $\Sigma_{ac}$ denote the support of the absolutely continuous spectrum of $H$ and let
$\Sigma_{s}$ denote the support of the singular spectrum of $H|\VV$, i.e. $H$ restricted to $\VV$.
The following theorem corresponds to \cite[Theorem~3]{KP}.
\begin{theorem}\label{th:subordinacy}
Let
 \begin{gather*}
 \Sigma'_{ac}\,:=\,\{E\,\not\in \overline{B_\infty}\,:\, \text{there is no subordinate solution}\;\} \\
 \Sigma'_s\,:=\,\{E\,\not\in \overline{B_\infty}\,:\,\text{$u_n=\smat{1&0} T_E(n) \smat{1\\0}$ is a subordinate solution}\;\} \\
 \Sigma'_0\,:=\,\{E\,\not\in \overline{B_\infty}\,:\,\text{$w_n=\smat{1&0} T_E(n) \smat{m\\-1}$ is subordinate for some $m\in\RR$}\}\;.
 \end{gather*}
  Then $\Sigma'_{ac}$ is an essential support of the a.c. spectrum of $H$ on $\RR\setminus \overline{B_\infty}$ and 
 $\Sigma'_s$ is an essential support of the singular spectrum of $H|\VV$ on $\RR\setminus \overline{B_\infty}$ which is optimal with respect to the Lebesgue measure. 
 This means $\mu_{0,ac}(\Sigma_{ac} \setminus (\Sigma'_{ac}\cup\overline{B_\infty}))=0$, 
 $\mu_{0,s}(\Sigma_{s} \setminus (\Sigma'_{s} \cup \overline{B_\infty}))=0$, $|\Sigma' _{ac} \setminus \Sigma_{ac}|=0$ and $|\Sigma'_s|=0$.
Here,  $\mu_{0,ac}$ and $\mu_{0,s}$ denote the a.c. and singular part of $\mu_0$ and $|\cdot|$ denotes the Lebesgue measure of a set. 
 Moreover, for Lebesgue almost all $E\in \Sigma'_0$  with subordinate solution $w_n$ where $w_{-1}=-1,\,w_0=m(E)\in\RR$ we find 
 \begin{equation}\label{eq-gr-limit}
 \lim_{\eta \downarrow 0} \,\langle \Phi_n\,|\,(H-(E+i\eta))^{-1}\,|\,\Phi_0\rangle\,=\,w_n\;.
 \end{equation}
 \end{theorem}

The proof of the characterization \cite[Theorem~1.1]{LaSi}
of an essential support of the a.c. spectrum depends on subordinacy theory
and a similar identity as Theorem~\ref{th:sigma_0}~(ii). So after establishing these theorems one might think that at least away from 
$\overline{B_\infty}$ one should have a similar characterization of 
an essential support of $\mu_{0,ac}$, the a.c. part of $\mu_0$.
Unfortunately, following the arguments in \cite{LaSi} does not quite give this result, unless one finds uniform constants $0<c<C$ such that\footnote{The bound $\psi_{E,n}\|\geq 1/\|\phi_n\|$ is trivial and already mentioned in \eqref{eq-rel-abpsi}}
$1\leq \|\psi_{E,n}\|\,\|\phi_n\|<C$ and $c<\|\psi_{E,n}\|\,\| \psi_{E,n-1}\|<C$
uniformly in $n$ and $E$ (locally). In that case the sets $\Sigma_\Phi$ and $\Sigma_\Psi$ as defined below are equal.
We find that \cite[Theorem~1.1 and 1.2]{LaSi} generalize to the presented situation in the following way.

\begin{theorem}\label{th:ac-necessary}
Let $\Sigma_{ac}$ be the support of the absolutely continuous spectrum of $H$ as before and
let $\boldsymbol{\Phi}_k=\diag(\|\phi_k\|\,,\,\|\phi_{k-1}\|)$ and $\boldsymbol{\Psi}_{E,k}=\diag(\|\psi_{E,k}\|\,,\,\|\psi_{E,k-1}\|)$.
Moreover, define the sets
 \begin{align*}
  \Sigma_\Phi\,&:=\,\left\{E\in\RR\,:\,\liminf_{n\to\infty}\;\frac1n\; \sum_{k=1}^n \left\|\,\boldsymbol{\Phi}_k^{-1}\, T_E(k) \,\right\|^2\;<\;\infty\,\right\}\;, \\
 \Sigma_\Psi\,&:=\,\left\{E\in\RR\,:\,\liminf_{n\to\infty}\;\frac{\sum_{k=1}^n \left\|\,\boldsymbol{\Psi}_{E,k}\, T_E(k) \,\right\|^2}
 {\sum_{k=1}^n \det(\bpsi_{E,k})} \;<\;\infty\,\right\}\;.
 \end{align*}
  Then, one has the following:
\begin{enumerate}[{\rm (i)}]
\item  For Lebesgue almost all $E\in \Sigma_{ac}$ we find $E\in \Sigma_\Phi$.
\item  For $E\in \Sigma_\Psi$ there is no subordinate solution at $E$ and thus for Lebesgue almost all $E\in  \Sigma_\Psi \setminus \overline{B_\infty}$ we find
$E\in \Sigma_{ac}$.

 \item Defining the transfer matrix from $k$ to $m$ by $T_E(k,m)=T_E(k) T_E(m)^{-1}$ we find for any fixed sequences $k_n,\, m_n$ and
 Lebesgue almost every $E\in\Sigma_{ac}$ that
 $$
 \liminf_{n\to\infty}\; \left\|\,\boldsymbol{\Phi}_{k_n}^{-1}\,T_E(k_n,m_n)\,\boldsymbol{\Phi}_{m_n}\,/\,\det(\boldsymbol{\Phi}_{m_n})\, \right\|\,<\,\infty\;.
 $$
 \end{enumerate}
\end{theorem}

Let us finally give some remark on a possible extension of the above theorems.
In general one may also want to consider operators as in \eqref{eq-H-prop} with some infinite dimensional fibers, i.e. allowing $s_n=\infty$ in the sense $\CC^\infty \cong \ell^2(\NN)$.
Then, $V_n$ should be a Hermitian operator on $\ell^2(\NN)$.
As long as all the $V_n$ have pure point spectrum without accumulation point, similar techniques apply, only the sets $A_n$ defined above 
are possibly countably infinite. However, more care must be taken when $V_n$ has some continuous spectrum or dense point spectrum. 
These cases might be interesting for further investigation.

\vspace{.2cm}

Now, for the singular spectrum we will essentially use the Simon-Wolff criterion in combination with 
\eqref{eq-gr-limit} in Theorem~\ref{th:subordinacy}. To show how this leads to the pure point or pure singular continuous spectrum, let us finish this section by proving the following general statement:

\begin{theorem}\label{th:purepoint}
Let $\boldsymbol{\varphi}=\varphi\oplus \bigoplus_{n\geq 1} \nul$ with $\varphi\in\VV_0$.
\begin{enumerate}[{\rm (i)}]
\item Assume that for Lebesgue almost every energy $E\in (a,b) $ we find a subordinate solution $w_{E,n}$ to \eqref{eq-eig-u} at $E$ such that
$$
\sum_{n=0}^\infty |w_{E,n}|^2\, \|\psi_{E,n}\|^2\,<\,\infty\;.
$$
Then, for Lebesgue almost every $c$, the operator $H+c\,|\boldsymbol{\varphi} \rangle \langle \boldsymbol{\varphi}|$ has pure point spectrum in $(a,b)$.
\item Let $(a,b) \cap B_\infty = \emptyset$ and assume that for Lebesgue almost all $E\in (a,b)$ there is a subordinate solution $w_{E,n}$ to
\eqref{eq-eig-u} at $E$ such that
$$
\sum_{n=0}^\infty |w_{E,n}|^2\, \|\psi_{E,n}\|^2\,=\,\infty\;.
$$
Then, for Lebesgue almost every $c$, the spectral measure of $H+c\,|\boldsymbol{\varphi} \rangle \langle \boldsymbol{\varphi}|$ 
at $\bphi$ {\rm (}i.e. $\mu_{0,\varphi}$ as in Theorem~\ref{th:sigma_0}{\rm )} is purely singular continuous in $(a,b)$.
\end{enumerate}
\end{theorem}

\begin{proof}
For part (i) first note that if $w_{n}$ is a subordinate solution of \eqref{eq-eig-u} for $E\not \in B_\infty$ such that $\sum_n |w_n|^2 \|\psi_{E,n}\|^2 < \infty$, then either
$w_{-1}=0$ and $\psi:=\bigoplus_n w_n \psi_{E,n}$ is an eigenvector of $H$ for the eigenvalue $E$, or
$w_{-1}\neq 0$ and $(H-E)\psi\,=\,- w_{-1}\,\Phi_0$. 
In the latter case one finds
$$
\sup_{\eta>0} \|(H-E-i\eta)^{-1}\,\Phi_0\|^2\,=\,\sup_{\eta>0}\,\left\|\frac{H-E}{H-E-i\eta}\;\;\psi/w_{-1} \right\|^2\,\leq\,\|\psi\|^2\,/\,|w_{-1}|^2\,<\,\infty\;.
$$
Using the Green's function identities in Lemma~\ref{lem:gr} one finds
$$
(H-z)^{-1} \boldsymbol{\varphi}=\left((V_0-z)^{-1}(\varphi-\phi_0\psi_{\bar z,0}^*\varphi) \right)\,\oplus\,\bigoplus_{n=1}^\infty \left( (\psi_{\bar z,0}^*\varphi) g_z(n,0) \psi_{z,n} \right)\,.
$$
where $g_z(n,m)= \langle \Phi_n|(H-z)^{-1}\Phi_m\rangle$ (Our scalar product $\langle\cdot|\cdot\rangle$ is linear in the second and anti-linear in the first component.) Comparing the general case with $\varphi=\phi_0$ (i.e. $\bphi=\Phi_0$)
we see that for $E\not \in \spec(V_0|\VV_0)\cup A_0$ it follows that
$$\sup_{\eta>0} \|(H-E-i\eta)^{-1} \bphi\, \|\,<\,\infty\;.$$
 As there are only countably many eigenvalues of $H$ and $V_0$, the latter equation is true for Lebesgue almost all $E\in(a,b)$.
Using the Simon-Wolff criterium, Theorem~2' in \cite{SW}, we find for $c\in L_1$, a set of full Lebesgue measure, that the spectral measure at
$\bphi$ of $H+c|\bphi\rangle \langle \bphi|$ in $[a,b]$ is pure point. Clearly, for $c\in L_2$, another set of full Lebesgue measure, $\varphi$ is 
in the cyclic space of $\Phi_0$ with respect to $V_0+c\; \varphi \varphi^*$. Hence, by Theorem~\ref{th:sigma_0} (Remark~\ref{rem:purepoint}~(i))
the operator $H+c|\bphi\rangle \langle \bphi|$ has pure point spectrum in $[a,b]$ for $c\in L_1\cap L_2$.

For part (ii) note that for $(a,b)\cap B_\infty=\emptyset$ we can apply Theorem~\ref{th:subordinacy} and \eqref{eq-gr-limit}.  We will use the notations as in Theorem~\ref{th:subordinacy}.
By assumption, the Lebesgue measure of $\Sigma'_{ac}\cap (a,b)$ is zero, therefore, there is no a.c. spectrum in $(a,b)$.

As the Lebesgue measure of $\Sigma'_0$ is zero, there is still a set of energies of full Lebesgue measure where
 $w_{-1}=-1$ and $w_n=\lim_{\eta\downarrow 0} g_z(n,0)$ is a subordinate solution with $\sum_n |w_n|^2 \|\psi_{E,n}\|^2\,=\,\infty$.
This implies 
$$
\sup_{\eta>0} \|(H-E-i\eta)\Phi_0\|^2\,=\,\infty \qtx{and} \sup_{\eta>0} \|(H-E-i\eta)^{-1} |\bphi \|^2\,=\,\infty\;
$$
for Lebesgue almost all energies $E\in [a,b]$. By the Simon-Wolff criterium this means that for Lebesgue almost every $c$ the spectral measure of
$H+c|\bphi\rangle \langle \bphi|$ at $\bphi$
is purely continuous in $(a,b)$. As there can be no a.c. spectrum in $(a,b)$ as mentioned above,  it has to be purely singular continuous.
\end{proof}

\section{Proof of Lemma~\ref{lem-sub-keyest} \label{sec:key-estimate-sub}}

The proof of part~(ii) of Theorem~\ref{th:purepoint} depends heavily on the subordinacy theory 
and in particular on \eqref{eq-gr-limit}.  
In fact, subordinacy theory is the key for parts (ii) to (iv) of Theorem~\ref{th:main}. 
As we mentioned above, the estimates given in Lemma~\ref{lem-sub-keyest} are crucial which we prove in this section. They may also be useful in other
circumstances. Recall that we want to show:
\begin{enumerate}[{\rm (i)}]
\item For any $E \in \RR\setminus B_\infty$ we have $ |a_{E+i\eta,n}-a_{E,n}|\,/\,\|\psi_{E,n}\|^2 \,\leq\,\eta$ and $\|\psi_{z,n}\|\,\leq\,\|\psi_{E,n}\|$.
\item If $(E-3\varepsilon\,,\,E+3\varepsilon) \cap B_\infty = \emptyset$ then we find a uniform constant $C>0$ independent of $n$ (and in fact of $V_n$) 
such that for all $0\leq \eta \leq \varepsilon$ one has
$
1\,\leq\, \|\psi_{E,n}\|\,/\,\|\psi_{E+i\eta,n}\|\,\leq\,C\;.
$
\end{enumerate}
One may note that these estimates are completely trivial in the one dimensional Jacobi matrix case\footnote{Indeed, in that case one simply has
$s_n=1$, $a_{z,n}=(z-v_n)/(|\phi_n|^2)$ and $\psi_{z,n}=1/(|\psi_n|^2)$}.

\begin{proof}[Proof of Lemma~\ref{lem-sub-keyest}]
For part (i) note that for $z=E+i\eta$
$$
|a_{z,n}-a_{E,n}|\,\geq\,\big|\im(a_{z,n})\big|\,=\,
\eta\,\|\psi_{z,n}\|^2
\,\geq\,
\eta\,\left| \frac{d}{dz} a_{z,n} \right|\,\geq\,\eta\,\left|\frac{d}{d\eta}\left( |a_{E+i\eta}-a_E|\,\right)\right|\,.
$$
Using the mean value theorem it follows that for any $\eta>0$ there is $0<\eta'<\eta$ such that 
$\left|\frac{d}{d\eta} (|a_{E+i\eta'}-a_E|)\right|\geq\left|\frac{d}{d\eta} (|a_{E+i\eta}-a_E|)\right|$ and hence the maximum derivative must be at $\eta=0$.
Thus, by the mean value theorem, 
$$\|\psi_{E+i\eta,n}\|^2 \leq \eta^{-1} |a_{E+i\eta}-a_E| \,=\, \left| \frac{d}{d\eta}\left( |a_{E+\eta',n}-a_{E,n}|\right)\right| \,\leq\,
\left| \frac{d}{dE} a_{E,n} \right|=\|\psi_{E,n}\|^2
$$ proving \eqref{eq-imp-est}.

For part (ii) note that the first inequality is proved in part (i) and we only need to worry about an upper bound for $\|\psi_{E,n}\|^2\,/\,\|\psi_{E+i\eta,n}\|^2$ uniformly in $V_n$ and $\phi_n$. Without loss of generality
we may assume $E=0$ by changing $V_n$ to $V_n-E$. Let us further use an orthonormal basis of eigenvectors in $\VV_n$ such that $V_n|\VV_n=\diag(-x_1,\ldots,-x_k,y_1,\ldots,y_l)$ where $x_i\geq 0$ and $y_j>0$. All these values are different
as $\phi_n$ is a cyclic vector. Moreover, we let $\phi_n=(\alpha_1,\ldots,\alpha_k,\beta_1,\ldots,\beta_l)^\top$ in this basis and let $a_i=|\alpha_i|^2>0$ and $b_j=|\beta_j|^2>0$.
By cyclicity, none of these values is zero.
Then we define
$$
f(x)\,:=\,\phi_n^* (x-V_n)^{-1} \phi_n\,=\, \sum_i \frac{a_i}{x_i+x}\,-\,\sum_j \frac{b_j}{y_j-x}\;.
$$
By assumption $(-3\varepsilon,3\varepsilon)\cap B_\infty =\emptyset$ and hence $f(x)\neq 0$ for $|x|<3\varepsilon$. Moreover, $f(x)$ is decreasing from $+\infty$ to $-\infty$ between single poles located at the $-x_i$ and $y_j$. 
So $f(x)$ has at most one pole inside $(-3\varepsilon,3\varepsilon)$ which we may assume to be $-x_1$ (the case of a pole at a different $-x_i$ or some $y_j$ inside $(-3\varepsilon,3\varepsilon)$ is completely analogue).
Thus, without loss of generality we may assume
$f(\delta)\geq 0$ for all $0<\delta\leq 3\varepsilon$ and $y_j>3\varepsilon$ for all $j$ and $x_i>3\varepsilon$ for $i\geq 2$.\\
\noindent Case 1: We have $x_1\in [0,\varepsilon]$.
Then $f(-3\varepsilon)\leq 0$ and $f(3\varepsilon)\geq 0$ and for any $0\leq \eta \leq 3\varepsilon$ we find
\begin{align}
& \frac{a_1}{4\varepsilon}\geq\frac{a_1}{x_1+3\varepsilon}\geq \sum_{j=1}^l \frac{b_j}{y_j-3\varepsilon}\,-\,\sum_{i=2}^k \frac{a_i}{x_i+3\varepsilon}\,\geq\,
\sum_{j=1}^l \frac{b_j y_j}{y_j^2+\eta^2}\,-\,\sum_{i=2}^k \frac{a_i x_i}{x_i^2+\eta^2} \\
& \frac{a_1}{2\varepsilon}\geq\frac{a_1}{3\varepsilon-x}\geq \sum_{i=2}^k \frac{a_i}{x_i+3\varepsilon}\,-\, \sum_{j=1}^l \frac{b_j}{y_j-3\varepsilon}\,\geq\,
\sum_{i=2}^k \frac{a_i x_i}{x_i^2+\eta^2}\,-\,\sum_{j=1}^l \frac{b_j y_j}{y_j^2+\eta^2}\;.
\end{align}
Combining both equations gives
$$
\left|\,\sum_{j=1}^l \frac{b_j y_j}{y_j^2+\eta^2}\,-\,\sum_{i=2}^k \frac{a_i^2 x_i}{x_i^2+\eta^2}\,\right|\,\leq\,\frac{a_1}{2\varepsilon}\,\leq\,\frac{a_1}{2x_1}
$$
and summing both equations and dividing by $6\varepsilon$ gives
$$
\frac{a_1}{8\varepsilon^2}\,\geq\,\sum_{i=2}^k \frac{a_i}{x_i^2-9\varepsilon^2}\,+\,\sum_{j=1}^l \frac{b_j}{y_j^2-9\varepsilon^2}\,>\,
\sum_{i=2}^k \frac{a_i}{x_i^2+\eta^2}\,+\,\sum_{j=1}^l \frac{b_j}{y_j^2+\eta^2}
$$
for any $0\leq \eta \leq \varepsilon$.
Using these estimates one finds
$$
\|\psi_{0,n}\|^2\,=\,
\frac{\sum_{i=1}^k a_i /{x_i^2}+\sum_{j=1}^l {b_j}/{y_j^2}}{\left(\sum_{i=1}^k {a_i}/{x_i}-\sum_{j=1}^l {b_j}/{y_j}\right)^2}\,\leq\,
\frac{a_1/x_1^2+a_1/(8x_1^2)}{\left(a_1/x_1 - a_1/(2x_1)\right)^2}\,\leq\,\frac{9}{2 a_1}
$$
and
\begin{align*}
\frac{1}{\|\psi_{i\eta,n}\|^2}\,&=\,
\frac{\left(\sum_{i=1}^k \frac{a_i x_i}{x_i^2+\eta^2}-\sum_{j=1}^l \frac{b_j y_j}{y_j^2+\eta^2}\right)^2}{\sum_{i=1}^k \frac{a_i }{x_i^2+\eta^2}+
\sum_{j=1}^l \frac{b_j }{y_j^2+\eta^2}}\,+\,
\sum_{i=1}^k \frac{a_i \eta^2}{x_i^2+\eta^2}+\sum_{j=1}^l \frac{b_j \eta^2}{y_j^2+\eta^2}\,\\
&\leq\, \frac{\left(\frac{a_1 x_1}{x_1^2+\eta^2}+\frac{a_1}{2\varepsilon}\right)^2}{\frac{a_1}{x_1^2+\eta^2}\,-\,\frac{a_1}{8\varepsilon^2}} \;+\;\eta^2 
\left(\frac{a_1}{x_1^2+\eta^2}+\frac{a_1}{8\varepsilon^2}\right) \,\leq\,
\frac{\frac{2a_1^2 x_1^2}{(x_1^2+\eta^2)^2}+\frac{a_1^2}{\varepsilon^2}}{\frac34 \,\frac{a_1}{x_1^2+\eta^2}}\,+\,\frac{9a_1}{8}\,\\
&\leq\,a_1\left(\frac{8}{3}\,\frac{x_1^2}{x_1^2+\eta^2}\,+\,\frac{4}{3}\,\frac{x_1^2+\eta^2}{\varepsilon^2}\,+\,\frac{9}{8} \right)\,\leq\,
a_1\,\left( \frac{16}{3}\,+\,\frac{9}{8}\right)
\end{align*}
where we used $\eta,x_1\,\leq\varepsilon$ at several places.
Both estimates together give the required uniform upper bound on $\|\psi_{0,n}\|\,/\,\|\psi_{i\eta,n}\|$ which also remains valid in the limiting case $x_1= 0$.\\
\noindent Case 2: $x_i>\varepsilon$ for all $i$, recall that also $y_j>3\varepsilon$ for all $j$.
We find for $x\in (0,\varepsilon)$ that
$$
|f'(x)|=\sum_{i=1}^k \frac{a_i}{(x_i+x)^2}\,+\,\sum_{j=1}^{l} \frac{b_j}{(y_j-x)^2}\,\geq\,\frac{1}{4}\left(\sum_{i=1}^k \frac{a_i}{x_i^2}+\sum_{j=1}^l \frac{b_j}{y_j^2} \right)\,=\,\frac{|f'(0)|}{4}\;.
$$
Combining this estimate with the fact that $f$ is decaying on $(0,\varepsilon)$ and $f(\varepsilon)>0$ 
we find $f(0)\geq\frac\varepsilon{4} \,|f'(0)| $. Using  $\left|\frac{a_i}{x_i+\eta^2/x_i}-\frac{a_i}{x_i}\right|\leq \varepsilon \frac{a_i}{x_i^2}$ 
coming from $\eta^2/x_i < \varepsilon$, we obtain 
$$
\left|\,\sum_i \frac{a_i x_i}{x_i^2+\eta^2}-\sum_j\frac{b_j y_j}{y_j^2+\eta^2} \,-\, f(0)\right|\;\leq\; \varepsilon |f'(0)|\;
$$
for $\eta\leq\varepsilon$.
With this estimate and the expressions for $\|\psi_{0,n}\|^{2}$ and $\|\psi_{i\eta,n}\|^{-2}$ as above we then find for $\eta\in [0,\varepsilon]$ that
$$
\|\psi_{0,n}\|^2\,=\,\frac{|f'(0)|}{(f(0))^2}\;,\qquad
\frac{1}{\|\psi_{i\eta,n}\|^2}\,\leq\, \frac{\left(f(0)+\varepsilon |f'(0)| \right)^2}{\frac12 |f'(0)|}\;+\;\varepsilon^2 |f'(0)|
$$
which using $\varepsilon |f'(0)| / f(0) \leq 4$ gives $\|\psi_{0,n}\|\,/\,\|\psi_{i\eta,n}\|\,\leq\,\sqrt{66} $.
\end{proof}



\section{The essential spectrum \label{sec:spec-H}}

\subsection{Relation of resolvents of $H$ and the Jacobi operators $\Jj_E$}

We first prove Proposition~\ref{prop-spec_ess} and will then apply it to the Anderson model $H_\lambda$ on $\Ab_\bfs$.
Let us introduce the following notations:
\begin{defini}
For a sequence $w=(w_n)_n$, $w_n\in \CC^{\ZZ_+}$ we define
$$
\psi_E \odot w\,:=\,\bigoplus_{n=0}^\infty (w_n \psi_{E,n}) \qtx{and} \phi\odot w\,:=\,\bigoplus_{n=0}^\infty (w_n\,\phi_n)\,.
$$
If the resulting direct sums do not give $\ell^2$ vectors, then one may still understand it formally as a collection of vectors $(w_n\phi_n)_n \in \bigtimes_n \CC^{s_n}$.
\end{defini}
A crucial part is now the following equation that is easy to verify,
\begin{equation}
 (\Hh-z)\,(\psi_z\odot w)\,=\,\phi\odot (\Jj_z(w))\,,
\end{equation}
with $\Jj_z$ as in \eqref{eq-def-Jj_z}. Another important point which is not hard to check is that $(\Hh-z)\psi=\phi\odot u$ for some sequence $u$ implies $\psi=\psi_z\odot w$,  $\Jj_z(w)=u$ and $w_n=\phi_n^*\psi(n)$.

\begin{proof}[Proof of Proposition~\ref{prop-spec_ess}]
Recall that we assumed $\|\phi_n\|<C$ for some $C$ uniformly in $n$.
First, assume $E\not\in \spec(H)$, then $(H-E)^{-1}$ exists as a bounded operator.
Then by uniform boundedness of $\|\phi_n\|$, the operator
$$
(\Rr w)_n\,:=\,\langle \Phi_n\,|\,(H-E)^{-1} (\phi\odot w)\,\rangle
$$
is bounded on $\ell^2(\ZZ_+)$. Moreover, it is not hard to check that $(H-E)^{-1} (\phi\odot w) = \psi_E \odot \Rr w$ and hence
$\phi\odot w = (H-E) (\psi_E\odot \Rr w) = \phi\odot \Jj_E(\Rr w)$. Hence, $\Jj_E \Rr\, w = w$ for all $w$ and therefore
$\Jj_E \Rr=\mathbbm{1}$, $\Rr^*\Jj_E=\mathbbm{1}$ and $\Rr^*=\Rr^* \Jj_E \Rr = \Rr$. Thus, $\Rr=\Jj_E^{-1}$ and $0\not\in\spec(\Jj_E)$.
Thus, we have shown $0\in \spec(\Jj_E)\Rightarrow E \in \spec(H)$.

Now let $0\not\in \spec(\Jj_E)$, i.e. $\Jj_E^{-1}$ exists and assume $E\not \in \spec(\Vv)$, $\Vv\in\bigoplus_{n=0}^\infty V_n$. Then, 
\begin{align}
 & (\Hh-E)(\Vv-E)^{-1}\,\psi\,=\,\psi\,-\, \phi\odot w(\psi) \quad\text{where} \\
 & w(\psi)_n\,=\,\phi_{n+1}^*(E-V_{n+1})^{-1}\psi_{n+1}\,+\phi_{n-1}^*(E-V_{n-1})^{-1}\psi_{n-1}\,.
\end{align}
It follows that $(\Hh-E)(\Aa_E\,\psi)=\psi$ for
\begin{align}
 \Aa_E \,\psi\,:=\,(\Vv-E)^{-1}\,\psi\,+\,\psi_E\odot \Jj_E^{-1}(w(\psi))\,.
\end{align}
Hence, $E$ is in the resolvent set and $\Aa_E=(\Hh-E)^{-1}$ if $\Aa_E$ defines a bounded operator on $\HH$.
Now, $(\Vv-E)^{-1}$ is bounded and $\|w(\psi)\|_2\leq 2 (\sup_n \|\phi_n\|) \|(\Vv-E)^{-1}\|\, \|\psi\|$.
Therefore, it is enough to show that $A_E(w)=\psi_{E}\odot \Jj_E^{-1}(w)$ is bounded from $\ell^2(\ZZ_+)$ to $\bigoplus_n \CC^{s_n}$.

By our assumptions, $\|\psi_{E,n}\|=\|(V_n-E)^{-1}\phi_n\| |a_{E,n}| \leq C |a_{E,n}|$ for some constant $C$.
Define the scalar product norm $\|v\|^2_{a_E}=\sum_n (1+a_{E,n}^{2})|v_n|^2$ for sequences $v=(v_n)_n$.
This norm is equivalent to the norm $\|v\|^2_{\Jj_E}=\|v\|^2_2+\|\Jj_E v\|_2^2$.
Therefore, the domain of $\Jj_E$ is exactly the 
set where $\| \cdot\|_{a_E}$ is bounded and $\Jj_E$ has an inverse on $\ell^2(\ZZ_+)$ (i.e. $0$ is in the resolvent set),
if and only if $\Jj_E$ is invertible as an operator from $\ell^2(\ZZ_+,\|\cdot\|_{a_E})$ to $\ell^2(\ZZ_+)$.
In particular, in this case one has a bound for the norm $\|\Jj_E^{-1}(w)\|_{a_E}$ in terms of $\|w\|_2$.
Hence, if $0\not\in\spec(\Jj_E)$, i.e. $\Jj_E^{-1}$ exists, then
$$
\big\|\psi_E\odot \Jj_E^{-1}(w)\big\|\,\leq\,C\,\|\Jj_E^{-1}(w)\|_{a_E}\,\leq\,C\|\Jj_E^{-1}\|_{2\to a_E}\,\| w\|_2\,.
$$
This shows boundedness of $A_E$, and hence of $\Aa_E$ and $E\not\in\spec(\Hh)$. Therefore we have shown $E\in\spec(H)\Leftrightarrow 0\in\spec(\Jj_E)$.
Now let $\Jj_E^{(n)}$ denote the restriction of $\Jj_E$ to $\ell^2(\{m\in\ZZ:m\geq n\}$ and let $H^{(n)}$ denote the restriction of $H$ to $\bigoplus_{m\geq n} \CC^{s_m}$. Then,
\begin{align*}
& 0\in\spec_\ess(\Jj_E) \;\Leftrightarrow\; 0\in \spec(\Jj_E^{(n)}) \text{ for all } n\in\ZZ_+\;\Leftrightarrow\;\\ & E\in \spec(H^{(n)}) \text{ for all } n\in\ZZ_+\;\Leftrightarrow\,E\in\spec_\ess(H).
\end{align*}
The first equivalence is a special case of the last equivalence when setting $s_n=1$ for all $n$ and the second equivalence was proved above\footnote{If $E\in \spec(\Vv)$ then we have only the implication '$\Rightarrow$' in the middle}, so let us consider the last claimed equivalence.
The direction '$\Leftarrow$' is clear as finite rank perturbations like formally changing $\phi_{n-1}$ to $\nul$, do not change the essential spectrum. Assume $E\in\spec(H^{(n)})$ for all $n$ and $E\not\in\spec_\ess(H)$. 
This would imply that $E$ is an isolated eigenvalue of finite multiplicity for any $H^{(n)}$.
As $E\not\in B_\infty\cup \spec(\Vv)$ and $\|\psi_{E,n}\|\geq 1/\|\phi_n\|>c>0$ uniformly, 
this means one finds two linear independent solutions $(u_n)_n$ and $(v_n)_n$ to \eqref{eq-eig-u} that are both $\ell^2$ at infinity.
This contradicts the invariance\footnote{The invariance can be easily checked, comes from the fact that the determinant of the transfer matrices is $1$} 
of the Wronskian $W=u_{n+1}v_n-v_{n+1}u_n$ that can not go to zero.
\end{proof}

\subsection{Essential spectrum of the Anderson model}

In this subsection we will prove Proposition~\ref{prop-spec-H}.
From now on (except for the appendix) we will consider the random operator $H_\lambda$ as in \eqref{eq-H_lb}.
Recall that for $H_\lambda$ we have $\phi_n=s_n^{-1/2}\, (1,1,\ldots,1)^\top$ and $\lambda\,V_n=\lambda\, \diag(v_{n,1},\ldots,v_{n,s_n})$
which replaces $V_n$ in the definitions of $a_{E,n}$, $\psi_{E,n}$, $T_{E,n}$ and $T_E(n)$, i.e.
\begin{equation}\label{eq-harm-entry}
a_{E,n}\,=\,(\phi_n^* (E-\lambda V_n)^{-1} \phi_n)^{-1}\,=\,\frac{1}{\frac1{s_n}\,\sum_{j=1}^{s_n} \frac{1}{E-\lambda v_{n,j}}}\;,
\end{equation}
\begin{equation}\label{eq-def-psi-T}
\psi_{E,n}=a_{E,n} (E-\lambda V_n)^{-1}\,\phi_n\qtx{and}
T_{E,n}=\pmat{a_{E,n} & -1 \\ 1 & 0}\;.
\end{equation}
These are now random objects as $(v_{n,j})_{n,j}$ is a family of
independent identically distributed random variables supported  in $[-1,1]$ with mean zero, $\EE(v_{n,j})=0$, and positive variance.
As before we let $\Vv=\bigoplus_n V_n$ so that $H_\lambda=\Aa_\bfs+\lambda \Vv$.
Similar as in \eqref{eq-def-Jj_z} we define the Jacobi operator with potentials $-a_{E,n}$ on $\ell^2(\ZZ_+)$ which is now a random, $\lambda$-dependent Jacobi operator
which we call $\Jj_{E,\lambda}$. Let us start with the following Lemma.
\begin{lemma}\label{lem-E-in-spec}
Let $\bfs=(s_n)_n$ be any sequence of positive integers characterizing the antitree $\Aa_\bfs$ and the operator $\Aa_\bfs$.
Let $E$ be an energy that is almost surely not in $B_\infty$ such that there exists a real $k$ so that for any $\varepsilon>0$ we have $\liminf_{n\to\infty} \PP(|a_{E,n}-2\cos(k)|<\varepsilon)>0$, formally
$$
E\in \RR\,,\quad\PP(E\in B_\infty)\,=\,0\qtx{and} \exists\,k\in\RR\;\forall\,\varepsilon\,>0\,:\, \liminf_{n\to\infty} \PP(|a_{E,n}-2\cos(k)|<\varepsilon)\,>0\,.
$$
Then almost surely $E\,\in\, \spec(H_\lambda)$.
\end{lemma}
\begin{proof}
 It is sufficient to show that we have $0\in\spec(\Jj_{E,\lambda})$ almost surely. For this we follow the usual proof as for the almost sure spectrum of the one-dimensional Anderson
 model by constructing a Weyl sequence.
 Let $\varepsilon>0$, then there is $\delta>0$ and $N>0$ such that  for all $n>N$ we have $\PP(|a_{E,n}-2\cos(k)|>\varepsilon)>\delta$.
 Therefore, with probability one we find arbitrarily long sequences where $a_{E,n}$ is $\varepsilon$-close to $2\cos(k)$, i.e. there is a (random) sequence $n_m,\,m\in\ZZ_+$ such that
 for all $m$ and all $1\leq j\leq m$ we have $|a_{E,n_m+j}-2\cos(k)|<\varepsilon$.
 Define 
$$
w_{m,n}\,=\,\begin{cases}
               0 & \quad \text{for $n>n_m+m$ or $n\leq n_m$} \\
              e^{ikn}  & \quad \text{for $n_m<n\leq n_m+m$}
              \end{cases} 
\quad \text{and consider $w_m=(w_{m,n})_n\in\ell^2(\ZZ_+)$}
$$
Then $\|w_m\|_2^2=m$ and $\|\Jj_{E,\lambda} w_m\|_2^2\,\leq\,2+\varepsilon^2 m$, thus $\limsup_{m\to\infty} \|\Jj_{E,\lambda} w_m\|_2\,/\,\|w_m\|_2\,\leq\,\varepsilon$ showing
$[E-\varepsilon,E+\varepsilon]\cap \spec(J_\lambda)\,\neq\,\emptyset$ almost surely, for any fixed $\varepsilon>0$.
Now take a sequence $\varepsilon_j\to 0$ and one obtains almost surely that $E\in\spec(\Jj_{E,\lambda})$.
\end{proof}
We also need the following
\begin{lemma}\label{lem-a-to-h}
Let $s_n>cn^\alpha$ for some $c>0, \alpha>0$ and let the single site distribution $\PP_v$ be compactly supported. 
We have almost surely for $n\to\infty$ that $a_{E,n}^{-1} \to h_{E,\lambda}^{-1}$ for all $E\not\in\lambda\supp(\PP_v)$, i.e.
 $a_{E,n}\to h_{E,\lambda}$ if $\EE(1/(E-\lambda v_{n,j}))\neq 0$ and $|a_{E,n}|\to \infty$ if $\EE(1/(E-\lambda v_{n,j}))= 0$.
\end{lemma}
\begin{proof}
Let the $X_i$ be independent identically distributed random variables with compactly supported distribution in $[-a,a]$ and zero expectation.
Let $Y=\frac1n \sum_{i=1}^n X_i$. Then
 $$
 \EE(Y^{2m})\,=\, \frac{1}{n^{2m}}\sum_{i_1,\ldots,i_{2m}=1}^n\!\!\!\!\! \EE(X_{i_1}\cdots X_{i_{2m}})\,\leq\,\frac{1}{n^{2m}} 
 \frac{(2m)!}{2^m m!}\sum_{j_1,\ldots,j_m=1}^n\!\!\!\!\!\EE(X^2_{j_1}\cdots X^2_{j_{m}})\,\leq\,
 \frac{(2m)!}{2^m m!} \,\frac{a^{2m}}{n^{m}}
 $$
where we use that $ \frac{(2m)!}{2^m m!}$ is exactly the number of pairings on the set $\{1,\ldots,2m\}$ and any unpaired index leads to zero expectation.
 Choosing $X_j=1/(E-\lambda v_{n,j})-h_{E,\lambda}^{-1}$ for $j=1,\ldots,s_n$ and using Markov's inequality this gives 
$$
\PP(|a_{E,n}^{-1}-h_{E,\lambda}^{-1}|>\varepsilon)\,\leq\,\frac{\EE(a_{E,n}^{-1}-h_{E,\lambda}^{-1})^{2m}}{\varepsilon^{2m}}\,\leq\,\frac{C}{\varepsilon^{2m} s_n^{2m}}\,\leq\,\frac{C'}{n^{m\alpha}}
$$
Choosing $m$ large enough the right hand side is summable over $n$. Hence, by Borel-Cantelli $\PP(|a_{E,n}^{-1}-h_{E,\lambda}^{-1}|>\varepsilon\,\text{infinitely often})\,=\,0$.
Taking a sequence $\varepsilon_j\to 0$ the akmost sure convergence follows for a single energy. Hence, we have
almost surely $a_{E,n}^{-1}\to h_{E,\lambda}^{-1}$ for all $E\in\QQ\setminus \lambda\supp(\PP_v)$. 
As $h_{E,\lambda}^{-1}$ and all $a_{E,n}^{-1}$ are continuous and monotone decreasing in $\RR\setminus\lambda \supp(\PP_v)$ the claim follows.
\end{proof}

\begin{proof}[Proof of Proposition~\ref{prop-spec-H}]
First note that one always has $\spec(H_\lambda)\subset\spec(\lambda\Vv)+[-2,2]=[-2,2]+\lambda\supp(\PP_v)$.
Consider first the case that the sequence $s_n$ is bounded, i.e. $s_n\,<\,N$ and let $E\in ([-2,2]\setminus\{0\})+\lambda \,\supp(\PP_v)$, i.e. $E=2\cos(k)+t$ where
$\cos(k)\neq 0$. Then  we find for $\varepsilon<2\cos(k)$ that
$$
\PP(|a_{E,n}-2\cos(k)|<\varepsilon)\,\geq\,\PP(\|\lambda V_n-t\,\one |<\varepsilon)\,>\,[\PP_v(\lambda^{-1}(t-\varepsilon,t+\varepsilon))]^N\,>0\,
$$
uniformly in $n$.
For the first inequality note that if $\|E-\lambda V_n-2\cos(k)\|<\varepsilon$ then each summand in \eqref{eq-harm-entry} lies between $1/(2\cos(k)+\varepsilon)$ and $1/(2\cos(k)-\varepsilon)$.
Moreover, the probability of $E\in B_\infty$ can only be positive for at most countably many energies, for the rest we can apply 
Lemma~\ref{lem-E-in-spec} showing $E\in \spec(H_{\lambda})$ almost surely.
Taking a countable dense set of energies we find $[-2,2]+\lambda\supp(\PP_v)\subset\spec(H_\lambda)$ almost surely, showing the first part.

\vspace{.2cm}

Now let $s_n>cn^\alpha$ and consider first $E\in\lambda\supp(\PP_v)$. We find with probably one a strictly increasing sequence $n_k$ such that
$|E-\lambda v_{n_k,1}|<1/k$ and $|E-\lambda v_{n_k,2}|<1/k$. Set $\psi_k=\delta_{n_k,1}-\delta_{n_k,2}$ i.e. $\psi_k(n)=\nul$ for $n\neq n_k$ and $\psi_k(n_k)=(1,-1,0,\ldots,0)^\top\in\CC^{s_{n_k}}$. Then, 
$$
\|(H_\lambda-E)\psi_k\| \,=\,\left\|\,(\lambda v_{n_k,1}-E)\,\delta_{n_k,1}\,-\,(\lambda v_{n_k,2}-E)\,\delta_{n_k,2}\,\right\|\,\leq\,2/k\,\to\,0\;.
$$
Hence, $\psi_k$ is an orthogonal Weyl sequence and $E\in \spec_\ess(H_\lambda)$ almost surely. Taking a countable dense set of energies we get $\lambda\supp(\PP_v)\subset \spec_\ess(H_\lambda)$ almost surely.

Let $E\not\in \lambda \supp(\PP_v)$. By Lemma~\ref{lem-a-to-h}, $\PP(|a_{E,n}-h_{E,\lambda}|<\varepsilon)\to 1$ for $n\to \infty$ whenever $|h_{E,\lambda}|<\infty$.
Hence, we can use Lemma~\ref{lem-E-in-spec} 
for a countable dense set of energies to find $\{E\,:\,|h_{E,\lambda}|\leq 2\}\,\subset\,\spec_\ess(H_\lambda)$ almost surely.\footnote{By monotonicity properties of $h_{E,\lambda}$ this set has no isolated points so it must be in the essential spectrum.}
By Lemma~\ref{lem-a-to-h} almost surely we find {\bf for all} $E$ with $|h_{E,\lambda}|>2$ some random $N=N(E)>0$ such that $|a_{E,n}|>2$ for $n\geq N(E)$. This shows $0 \not \in \bigcup_{n\geq N} [-2+a_{E,n},2+a_{E,n}]$ which includes the spectrum of $\spec(\Jj_E^{(N)})$ and hence $0\not\in\spec_\ess(\Jj_E)\subset\spec(\Jj^{(N)}_E)$. Thus, $\{E\,:\,|h_{E,\lambda}|>2\} \cap \spec_\ess(H_\lambda)=\emptyset$ almost surely.
\end{proof}


\section{Harmonic mean estimates \label{sec:keyestimate}}

Recall that for $E\in I_\lambda$ we either have that for all $v_{n,j}$, $E>\lambda v_{n,j}$ or that for all $v_{n,j}$, $E<\lambda v_{n,j}$.
Moreover, we have $|h_{E,\lambda}|<2$, where
\begin{equation}\label{eq-def-k}
h_{E,\lambda}=[\EE(1/(E-\lambda v_{n,j}))]^{-1}\,=:\,2\cos(k_{E,\lambda})\,.
\end{equation}
The latter equation is the definition of $k_{E,\lambda}\in(-\pi,\pi)$. 
In order to use the machinery in \cite{KLS} we need some explicit moment estimates.
Therefore, let us start with the following general theorem:
\begin{theorem} \label{lem:harmonic}
Let $X_j\in [a,b]$, $0<a<b$, be independent identically distributed random variables. 
Define the harmonic mean $M_n$ and the harmonic average $h$ by
$$
M_n \,:=\, \frac1{\frac1n \sum_{j=1}^n \frac1{X_j}}\;,\quad
h\,:=\,\frac1{\EE(1/X_j)}\;.
$$
Moreover, define the following moments of the centered random variable
$1/X_j-1/h$ 
$$
\sigma_m:=\EE((1/X_j - 1/h)^m)
$$
and note that $\sigma_1=0$ and $\sigma_2$ is the variance.
Then, there exists a continuous function $C=C(a,b,h,\sigma_2,\sigma_3)$ such that uniformly in $n$,
\begin{equation}\label{eq-exp-M_n}
 0< \EE(M_n-h) \,\leq\, \frac{b\,h^2\, \sigma_2}{n}\;,\qquad
 \left| \EE(M_n-h)- \frac{h^3 \sigma_2}{n}\right|\,\leq\,\frac{C}{n^2}\;.
\end{equation}
\begin{equation}\label{eq-exp-M_n-2}
\frac{a^2\,h^2\,\sigma_2}{n}\,\leq\,\EE((M_n-h)^2)\,\leq\,\frac{b^2\,h^2\,\sigma_2}{n}
\,,\qquad
\left|\,\EE((M_n-h)^2)-\frac{h^4 \sigma_2}{n}\,\right|\,\leq\, \frac{C}{n^2}\;.\qquad
 \end{equation}
Moreover, for the higher moments we find
\begin{equation}\label{eq-exp-M_n-3}
 \left|\,\EE((M_n-h)^3)\,\right|\,\leq\, \frac{C}{n^2}\,,\qtx{and}
 \EE((M_n-h)^{2m})\,\leq\,\frac{(2m)!\,h^{2m}\,b^{2m} }{2^m \,m!\,a^{2m}}\,\frac{1}{n^m}\quad\text{for $m\geq 2$}\,.
\end{equation}
\end{theorem}

\begin{proof}
Clearly, as $a>0$, we find $M_n \in [a,b]$ and $M_n\leq\frac1n \sum_{j=1} X_j$ as well as $h<\EE(X_j)$ by convexity.
 Let $Y=1/M_n-1/h=\frac1n \sum_{j=1}^n (\frac{1}{X_j} - 1/h)$, then one finds $\EE(Y)=0$,\,
 $\EE(Y^2)=\sigma_2 / n$, and $\EE(Y^3)=\sigma_3/n^2$. 
 A similar calculation as in the proof of Lemma~\ref{lem-a-to-h} shows
 $
 \EE(Y^{2m})\leq\frac{(2m)!}{2^m m!} \,\frac{1}{a^{2m}\, n^{m}}. $
 Moreover, one finds $M_n=h-hYM_n$ and therefore
 \begin{equation}\label{eq-exp-M}
 M_n-h = -hYM_n = -h^2Y+h^2Y^2M_n=-h^2Y+h^3Y^2-h^3Y^3M_n\;.
\end{equation}
As $M_n\in[a,b]$ we can estimate $\EE(h^2Y^2 M_n)\leq h^2 b\,\sigma_2/n$ and
$$
\left| \EE(Y^3M_n)\right|\,\leq\, \left|\EE(hY^3)\right|+\left|\EE(hY^4M_n) \right|\,\leq\,
\frac{h(|\sigma_3|\,+\,3ba^{-4})}{n^2}
$$
which with \eqref{eq-exp-M} (using the second-last and last term) gives \eqref{eq-exp-M_n}.
Taking powers of \eqref{eq-exp-M} and using similar estimates lead to \eqref{eq-exp-M_n-2} and 
\eqref{eq-exp-M_n-3}.
For the general moment bound note that \mbox{$1-YM_n$}$=M_n/h \in [a/h,b/h]$, so $|1-YM_n|\leq b/h$ and hence
$\EE((M_n-h)^{2m})\,\leq\, (b/h)^{2m}\, \EE((h^2 Y)^{2m})\,=\, h^{2m}\, b^{2m}\, \EE(Y^{2m})\;.$
\end{proof}

\vspace{,2cm}

Clearly, for $E\in I_\lambda$ all the transfer matrices $T_{E,n}$ exist and we may define the family of independent random variables
\begin{equation}\label{eq-def-x_n}
x_n:= \frac{a_{E,n}-h_{E,\lambda} }{\sin(k_{E,\lambda})}\,=\, 
\left[\,\left(\frac1{s_n}\sum_{j=1}^{s_n} \frac{1}{E-\lambda v_{n,j}}\right)^{-1}\,-\,h_{E,\lambda}\right]\,/\,\sin(k_{E,\lambda})\;.
\end{equation}
The introduced factor $\sin(k_{E,\lambda})$ will simplify some formulas later. 

For many calculations we will fix some $\lambda$ and $E\in I_\lambda$, therefore, we will now often 
omit the indices $E$ and $\lambda$ in future calculations and use 
$$
h=h_{E,\lambda}\,,\quad k=k_{E,\lambda}\,,\quad \sigma^2=\sigma^2_{E,\lambda}=\Var(1/(E-\lambda v_{n,j}))\,.
$$
Note that all this quantities depend continuously on $(E,\lambda)$ for $E\in I_\lambda$.
Without loss of generality we may assume $E>0$ which also corresponds to $E> h>0$, $k>0$ and $\EE(x_n)>0$. The considerations for $E<0$ are completely analogue.
Then, in the notations of the above theorem we find that $\sin(k)\,x_n$ corresponds to $M_{s_n}-h$, the $X_j$ correspond to $E-\lambda v_{n,j}$, 
$a=E-\lambda v_+$ and $b=E-\lambda v_-$, and we have
\begin{align}\label{eq-x-bounds}
&x_n\,\in\,\left[\frac{E-\lambda v_+}{\sin(k)}\,,\,\frac{E-\lambda v_-}{\sin(k)}\right], \quad 0\,<\;\EE(x_n)\,=\,\frac{h^3\,\sigma^2}{\sin(k)\,s_n}+\Oo(s_n^{-2})\,,\\ \label{eq-x-bounds2}
&\EE(x_n^2)\,=\, \frac{h^4\,\sigma^2}{\sin^2(k)\,s_n}+\Oo(s_n^{-2})\,,\quad
\EE(x_n^3)\,=\,\Oo(s_n^{-2})\,,\quad \EE(x_n^{2m})\,=\,\Oo(C_m\,s_n^{-m})\;.
\end{align}
The error terms $\Oo(s_n^{-2})$ and $\Oo(C_m s_n^{-m})$ mean that the absolute value of the reminder terms are bounded by
$Cs_n^{-2}$ and $C_m s_n^{-m}$, respectively, where $C=C_{E,\lambda}$ and $C_m=C_{m,E,\lambda}$ depend continuously on $(E,\lambda)$ for $E\in I_\lambda$.
In particular, the error terms are uniform in $E$ on compact subsets of $I_\lambda$.
Note that from Lemma~\ref{lem-a-to-h} we get for $E\in I_\lambda$

\begin{lemma}\label{lem-conv}
Let $s_n>cn^\alpha$  for some $\alpha>0$. Then, for $\PP$-almost all $\omega\in\Omega$ we have $$\lim_{n\to\infty} \,x_n\,=\,\lim_{n\to\infty} \,x_n(\omega)\,=\,0\;.$$
\end{lemma}


\section{Modified Pr\"ufer variables and results \label{sec:pruefer}}

We established that for $E\in I_\lambda$ and $s_n\to\infty$ the random transfer matrices behave like in the case of a one-dimensional operator with 
random potentials of decaying
variance. The conclusions can now be obtained in a similar way as in \cite{KLS} with some slight differences.
One minor difference is the fact that $\EE(x_n)$ is not zero and depends on $n$. Another difference for the case of the pure point spectrum is the fact that unlike
in \cite{KLS} the support of the distribution of $x_n$ is not getting smaller in $n$. But this input can be replaced by Lemma~\ref{lem-conv}.

Let us now briefly mention the appropriate basis change for the transfer matrices that leads to the modified Pr\"ufer variables that were also used
in \cite{KLS, LaSi}. By the definition of transfer matrices $T_{E,n}$ in \eqref{eq-def-psi-T}, the definition of $k=k_{E,\lambda}$ as in \eqref{eq-def-k} and 
$x_n$ as in \eqref{eq-def-x_n} 
one obtains
$$
M T_{E,n} M^{-1}= \pmat{1 & x_n \\ 0 & 1} \pmat{\cos(k) & -\sin(k) \\ \sin(k) & \cos(k)}\qtx{where}
M:=\pmat{1 & -\cos(k) \\ 0 & \sin(k)}\;.
$$
Next define the modified Pr\"ufer variables $\theta_n=\theta_n(\theta) \in \RR \mod  2\pi$ and $R_n=R_n(\theta,E)\in\RR_+$ by
\begin{equation}\label{eq-def-R-th}
R_n \vec{u}_{\theta_n} = M^{-1} T_E(n) M \vec{u}_\theta\;\qtx{where} \vec{u}_\theta=\pmat{\cos(\theta)\\ \sin(\theta)}\;,
\end{equation}
then some simple calculation shows
$
R_{n+1} \vec{u}_{\theta_{n+1}}\,=\,R_n \smat{\cos(\theta_n+k)+ x_n\sin(\theta_n+k)\\ \sin(\theta_n+k)}
$
which gives 
\begin{equation} \label{eq-rec-R}
R_{n+1}^2=R_n^2\left(1+x_n\sin(2\,\bar\theta_n)+x_n^2\sin^2(\bar\theta_n)\right)\;,\quad
\cot(\theta_{n+1})=\cot(\bar\theta_n)\,+\,x_n\;
\end{equation}
for $\bar\theta_n:=\theta_n+k$.
Note that the random variables $R_n,\,\theta_n,\,\bar \theta_n$ depend on the starting value $\theta=\theta_{0}$ and $x_0,\ldots,x_{n-1}$ and are therefore
independent of $x_n$.
By equivalence of norms one finds for any two linear independent angles $\vec{u}_\theta,\,\vec{u}_{\theta'}$ some positive constants $c,C$ such that
$$
c\,\max(R_n(\theta), \, R_n(\theta'))\,\leq\,\|T_E(n)\|\,\leq\, C\,\max(R_n(\theta),\,R_n(\theta'))\,.
$$
Therefore, it will be enough to study the $R_n$ in order to investigate $\|T_E(n)\|$.

\subsection{The absolutely continuous spectrum \label{sec:ac-spectrum}}
 
Now assume $\sum_{n=0}^\infty s_n^{-1}\,<\,\infty$. By \eqref{eq-x-bounds} and \eqref{eq-rec-R} we find for any starting angle $\theta=\theta_{0}$ that
\begin{align*}
\EE(R_{n+1}^4)\,&\leq\,\EE(R_n^4)+\EE(2\sin(2\bar\theta_n)R_n^4)\,\EE(x_n)+ \EE(R_n^4)\,\EE(3x_n^2+2|x_n|^3+x_n^4) \\
&\leq\,\EE(R_n^4)\,\left(1+2|\EE(x_n)|+\EE(4x_n^2+2x_n^4)\right)\,\leq\,
\EE(R_n^4)\,\left(1+C_{E}\,s_n^{-1}\right)
\end{align*}
where the bound $C_{E}$ can be chosen continuously in $E\in I_\lambda$.
As $C_{E} s_n^{-1}$ is summable it follows 
$$
\sup_{n}\,\EE(\|T_E(n)\|^4)\,\leq\, C \sup_{n} \sup_\theta \EE(R_n^4(E,\theta))  \leq
C\, \prod_{n=0}^\infty \left(1+\,C_{E}\,s_n^{-1}\right)\,<\,\infty
$$
where the bound is uniform in $E$ on compact sets $E\in[a,b]\subset I_\lambda$. 
Using Fatou's  lemma and Fubini 
we realize that $\PP$-almost surely
$$
\liminf_{n\to\infty} \,\int_a^b \|T_n(E)\|^4\,dE\,<\,\infty
$$
which by Theorem~\ref{th:sigma_0-ac} used for any $[a,b]\subset I_\lambda$, $a,b\in\QQ$, implies that the spectrum of $H_\lambda$ (restricted to the space $\VV$) 
is almost surely purely absolutely continuous in $I_\lambda$. For the proof of Theorem~\ref{th:sigma_0-ac} see Section~\ref{sec:sigma_0}.
As $|E|>\lambda\geq\|\lambda V_n\|$ for $E\in I_\lambda$ we also see that there are no eigenvalues in $I_\lambda$ with eigenvectors in $\VV^\perp$.
This proves Theorem~\ref{th:main}~(i).

\subsection{The singular spectrum \label{sec:pp-spectrum}}

In this section we want to prove Theorem~\ref{th:main} parts (ii), (iii) and (iv).
All will be based on the following observation which is a variant of \cite[Theorem~8.2 and Lemma~8.8]{KLS}. 
We set $\alpha=d-1$ where $d$ is the growth-rate dimension\footnote{Note that $\alpha$ corresponds to $2\alpha$ in \cite{KLS}}.

\begin{theorem}\label{th-lim-log(T)}
Assume $\lim_{n\to\infty} s_n^{-1} \,n^{\alpha}=c>0$ for some $0<\alpha\leq 1$, $E\in I_\lambda$ such that
$k=k_{E,\lambda} \not \in \,\pi/4\,\ZZ$.
Then one has almost surely 
\begin{equation}\label{eq-log(T)}
\lim_{n\to \infty} \frac{\log \| T_E(n)\,\|}{\sum_{j=0}^n s_j^{-1}}\;=\;
\lim_{n\to \infty} \frac{\log \| T_E(n)\,\|}{ \sum_{j=1}^n c\,j^{-\alpha}}\;=\;
\frac{h_{E,\lambda}^4\,\sigma_{E,\lambda}^2}{8\,\sin^2(k_{E,\lambda})}\,=\,\gamma_{E,\lambda}\;.
\end{equation}
Moreover, almost surely, there is a real, subordinate solution $w_{E,n}$, i.e.
$\smat{w_{E,n}\\ w_{E,n-1}}=T_E(n)\smat{w_{E,0} \\ w_{E,-1}}$, such that
\begin{equation}\label{eq-exist-subord}
\lim_{n\to\infty} \frac{\frac12\log(|w_{E,n}|^2+|w_{E,n-1}|^2)}{\log \| T_E(n) \|}\,=\,-1\;.
\end{equation}
\end{theorem}
\begin{rem}
The first equation is trivial provided the limit exists as $\left(\sum_{j=0}^n s_j^{-1}\right)\,/\,\left(\sum_{j=1}^n j^{-\alpha} \right)\,\to\,c$ for $n\to\infty$.  Moreover, note that for $\alpha<1$ we have $\sum_{j=1}^n j^{-\alpha}\,\sim\,(1-\alpha)^{-1}\,n^{1-\alpha}$ whereas
for $\alpha=1$, $\sum_{j=1}^n j^{-1}\,\sim\,\log(n)$.
The last equation follows from \eqref{eq-def-gamma} and \eqref{eq-def-k}.
\end{rem}

The proof can be done using the techniques of \cite{KLS}.
Differences are firstly that $\EE(x_n)$ varies with $n$ which will give some additional oscillatory term to take care of and secondly that the support of the distribution of $x_n$ is not getting smaller which can be replaced in the proofs by Lemma~\ref{lem-conv}.
For convenience of the reader, some more details are carried out in Appendix~\ref{app-estimate}.

\vspace{.2cm}

\begin{proof}[Proof of Theorem~\ref{th:main}] 
Part (i) is already proved above.
The remaining proof is based on subordinacy theory\footnote{Theorem~\ref{th:subordinacy}, for the proof see Section~\ref{sub:subordinacy}} mainly in the form of
Theorem~\ref{th:purepoint} which is proved at the end of Section~\ref{sub:spectrum}.

For $\alpha=0$ the formula \eqref{eq-log(T)} does not hold. However, if $s_n/n^0=s_n\to C$ then $s_n$ is constant for large $n$ and then
the corresponding transfer matrices are independent, identically distributed. Therefore, by standard arguments, one has a positive Lyapunov exponent and
a limit as in \eqref{eq-log(T)} does exist, only the formula on the right hand side is not valid.
It is also well known that a solution $w_{E,n}$ satisfying \eqref{eq-exist-subord} will exist almost surely (cf. \cite[Theorem~8.3]{LaSi}).
Note that for $E\in I_\lambda\cap\RR_+$ one also has that $\|\psi_{E,n}\|= \|(E-\lambda V_n)^{-1} \phi_n\|\,/\,|\phi_n^*(E-\lambda V_n)^{-1}\phi_n|\leq 
\frac{E+\lambda v_+}{E-\lambda v_+}<\infty$, and a similar bound holds for $E\in I_\lambda\cap \RR_-$.
Hence, the solutions $w_{E,n}$ of \eqref{eq-eig-u} are subordinate in the sense of Definition~\ref{def:subordinate}.

By a Fubini argument we find for $0\leq\alpha\leq 1$ almost surely such subordinate solutions for Lebesgue-almost all $E\in I_\lambda$.
By Theorem~\ref{th:subordinacy} this implies that there is no a.c. spectrum in $I_\lambda$ proving Theorem~\ref{th:main}~(ii).

\vspace{.2cm}

For $0<\alpha<1$ we find 
$
\lim_{n\to\infty} n^{\alpha-1}\,{\frac12\log\left(|w_{E,n}|^2\,+\,|w_{E,n+1}|^2\,\right)}\,=\,-c\,\gamma_{E,\lambda}$
and $\bigoplus_n w_{E,n} \psi_{E,n}$ is an $\ell^2$ vector where the decay rate is given by the decay rate of
the sequence $w_{E,n}$.
Under assumption (A2) ($\PP_v$ is absolutely continuous), we can use
Theorem~\ref{th:purepoint}~(i) with the basis vector\footnote{$\delta_{0,j}=\psi$ is defined by $\psi(n)=\delta_{0,n} e_j$, $e_j$ the $j$-th basis vector of $\CC^{s_0}$, the Dirac notation would be $|0,j\rangle$)}  $\bphi=\delta_{0,j}$
to obtain the almost sure pure point spectrum in $I_\lambda$ (see also Remark~\ref{rem:purepoint}~(i)).
Finally, the almost sure decay rate of the subordinate solutions gives the decay rate of the Green's functions and hence also the decay rate of the
eigenfunctions almost surely (cf. \cite[Theorem~9]{SW}).
By Proposition~\ref{prop-spec-H} the point spectrum also has to be dense (almost surely), finishing the proof of Theorem~\ref{th:main}~(iii).

\vspace{.2cm}

For $d=2$ or $\alpha=1$ we find
$
\lim_{n\to\infty} \frac12\log\left(|w_{E,n}|^2\,+\,|w_{E,n+1}|^2\,\right)\,/\,\log(n)\,=\,-c\,\gamma_{E,\lambda}
$
for the subordinate solutions $w_{E,n}$. However, the vector
$\bigoplus_n w_{E,n} \psi_{E,n}$ is in $\ell^2(\Ab_\bfs)$ only for $E\in I_\lambda \setminus J_\lambda$. If $E\in J_\lambda$ it is not in $\ell^2(\Ab_\bfs)$.
Again, using Fubini, assumption (A2), Theorem~\ref{th:purepoint}~(ii) and similar arguments as above we obtain the following:
The spectrum is almost surely pure point with the corresponding decay of the eigenfunctions in $I_\lambda\setminus J_\lambda$, and the spectral measures 
at $\delta_{0,j}$ (for all $j$) and at $\Phi_0=\sum_j \delta_{0,j}/\sqrt{s_0}$ are almost surely singular continuous in $J_\lambda$ for all $j=1,\ldots,s_0$.
Note that possible eigenvectors in $\VV^\perp$ (cf. Proposition~\ref{prop-spec-decom}) are outside $J_\lambda$ and 
$J_\lambda\cap B_\infty = \emptyset$. 
By Theorem~\ref{th:sigma_0} we therefore obtain that the spectrum of $H_\lambda$ is almost surely singular continuous in $J_\lambda$, finishing the proof of Theorem~\ref{th:main}~(iv).
\end{proof}

\section{The operator $\Pp_r\,\Delta_d\,\Pp_r\,+\lambda\Vv$ on $\ZZ^d$ \label{sec:PDP}}

Finally let us prove Corollary~\ref{coro:main}.
On $\ZZ^d$ the adjacency operator is given by 
$
\Delta_d\, \psi\,(x)\,=\,\sum_{y:y\sim x} \psi(y)$ where $y\sim x$ means that $y$ is a nearest neighbor of $x$, i.e. $\|x-y\|_1=1$.
Let 
$$S_n:=\{x\,\in\,\ZZ^d\,:\,\|x\|_1=n\}\;,\quad s_n=\#(S_n) \qtx{for} n\in \ZZ_+$$ and note that $\langle x\,\Delta_d|y\rangle\neq 0$ can only happen if
the difference of $\|x\|_1$ and $\|y\|_1$ is one.
Therefore, there are $s_n\times s_{n-1}$ matrices $D_n$ such that using spherical coordinates
$\psi=\bigoplus_n\psi(n)\in\bigoplus_n \CC^{s_n}\cong\bigoplus_n \ell^2(S_n)=\ell^2(\ZZ^d)$ one has
$$
(\Delta_d\,\psi)(n)\,=\,D_{n+1}^*\psi(n+1)\,+\,D_n\,\psi(n-1)\;
$$
where the entries of $D_n$ are zero or one giving the edges from $S_{n-1}$ to $S_n$. 

\begin{lemma}
 $\Pp_r \Delta_d \Pp_r - d\,\Aa_\bfs$ is trace-class on the Hilbert space $\bigoplus_{n=0}^\infty \CC^{s_n}$\,.
\end{lemma}
\begin{proof}
Let $\phi_n=1/\sqrt{s_n} (1,\ldots,1)^\top \in \CC^{s_n}$ as above, the radial projection is given by
$(\Pp_r \psi) (n)=\phi_n \phi_n^*\,\psi(n)$ or in Dirac notation $\Pp_r=\sum_n |\Phi_n\rangle\langle\Phi_n|$ where $\Phi_n=P_n \phi_n$ is the natural embedding of $\phi_n$ into $\ell^2(\ZZ^d)$.
Therefore,
$$
(\Pp_r\,\Delta_d\,\Pp_r\,\psi)(n)\,=\,\phi_n\left(a_{n+1}\,\phi_{n+1}^*\psi(n+1)\,+\,a_n\,\phi_{n-1}^*\,\psi(n-1) \right)\qtx{with} a_n=\phi_n^* D_n \phi_{n-1}\;.
$$
Now let $\alpha_n$ be the total number of edges between $S_{n}$ and $S_{n+1}$, then
$$
\alpha_n=\sum_{j=1}^{s_{n}} \sum_{k=1}^{s_{n-1}} (D_n)_{jk} \qtx{and}
a_n\,=\,\frac{1}{\sqrt{s_n s_{n-1}}}\,\alpha_{n-1}\;.
$$
To estimate $\alpha_n$ define for $k=0,1,\ldots,d$ 
$$
S_{n,k}=\big\{x\in (\ZZ)^d\,:\,\|x\|_1=n\;\text{and}\;\#\{j:x_j=0\}=k \big\}\;,\quad s_{n,k}=\#(S_{n,k})\;.
$$
In words, $S_{n,k}\subset S_n$ is the subset of vectors where exactly $k$ entries are zero. Clearly, $s_n=\sum_k s_{n,k}$.
Each vector $x=(x_1,x_2,\ldots,x_d)\in S_{n,0}$ can be mapped to an increasing sequence of $d$ positive integers $(|x_1|,|x_1|+|x_2|,\ldots)$ and a 
vector of signs $(\sgn(x_1),\ldots,\sgn(x_n))$. This is a bijection, therefore one obtains $s_{n,0}=2^d\binom{n-1}{d-1}$.
Similar reasoning shows for $n\geq d > k$ (note $s_{n,d}=0$ for $n>0$)
$$
s_{n,k}\,=\,\binom{d}{k}\, 2^{d-k}\,\binom{n-1-k}{d-1-k}\;.
$$
Moreover, each $x\in S_{n,k}$ has exactly $d+k$ edges to $S_{n+1}$ and $d-k$ edges to $S_{n-k}$, thus $\alpha_n=\sum_k (d+k) s_{n,j}=\sum_j (d-k) s_{n+1,j}$.
Up to errors of order $\Oo(n^{d-3})$ we find
$$
s_{n,0}\,=\,\frac{2^d\,n^{d-1}}{(d-1)!}-\frac{2^{d-1} d\,n^{d-2}}{(d-2)!}\,+\,\Oo,\quad s_{n,1}\,=\,\frac{2^{d-1} d\,n^{d-2}}{(d-2)!}\,+\,\Oo,\quad
s_n\,=\,\frac{2^d\,n^{d-1}}{(d-1)!}\,+\,\Oo
$$
and
$$
\alpha_n\,=\,\sum_{k=0}^{d-1} (d+k)\,s_{n,k}\,=\,\frac{d\, 2^{d}}{(d-1)!}\,\left(n^{d-1}+\frac12\,(d-1)n^{d-2}\right)\,+\,\Oo(n^{d-3})\,.
$$
This gives
$$
a_{n+1}^2\,=\,\frac{\alpha_n^2}{s_n\,s_{n+1}}\,=\,\frac{d^2\,(n^{2d-2}+(d-1)\,n^{2d-3})+\Oo(n^{2d-4})}{n^{2d-2}+(d-1)\,n^{2d-3}+\Oo(n^{2d-4})}\,=\,d^2+\Oo(n^{-2})
$$
which implies $a_n=d+\Oo(n^{-2})$ and
\begin{eqnarray*}
\|\Pp_r\Delta_d\Pp_r\,-\,d\,\Aa_\bfs\|_1\,=\,\left\| \sum_{n=1}^\infty (a_n-d) \left(|\Phi_{n}\rangle\langle\Phi_{n-1}|+|\Phi_{n-1}\rangle\langle\Phi_n| \right)\right\|_1\,\leq\,2\sum_{n=0}^\infty |a_n-d|\,<\,\infty\,.
\end{eqnarray*}
Thus, the difference is trace-class.
\end{proof}

\begin{proof}[Proof of Corollary~\ref{coro:main}]
The cases $d\leq 2$ follow immediately as a trace-class perturbation does not change the a.c. spectrum.
For the purity of the a.c. spectrum in the case $d\geq 3$ note that when exchanging $\Aa_\bfs$ by $\frac1d \Pp_r \Delta \Pp_r$ we have to change $\phi_n$ by $\phi'_n=(1+b_n)\phi_n$ with $b_n>-1$ 
where $(1+b_n)(1+b_{n-1})=a_n/d$. We have a free choice of $b_0>-1$ and it determines all other $b_n$.
By the formulas above it is clear that $a_n^2$ is a rational function in $n$ converging to $1$ and so it has to be monotone for large $n$.
As $a_{n+1}/a_{n}= (1+b_{n+2})/(1+b_n)$ we find that the sequences $(1+b_{2n})$ and $(1+b_{2n+1})$ are either both increasing or both decreasing (for large $n$) and converging.
Adjusting $b_0$ we can arrange for $b_n\to 0$ and then for some $n>N$ all $b_n$ have the same sign and one must have $b_n=\Oo(n^{-2})$ as $a_n/d=1+\Oo(n^{-2})$.
The upper right entries of the transfer matrices would change to  $a'_{E,n}=a_{E,n}\,/\,(1+b_n)^2\,=\,a_{E,n} (1+\Oo(n^{-2})$. Therefore,
in the proof of the pure a.c. spectrum in Section~\ref{sec:ac-spectrum}, $x_n$ would change to $x_n'=x_n(1+\Oo(n^{-2}))+\Oo(n^{-2})$.
As $n^{-2}$ is summable, we still obtain $\liminf_n \int_a^b \|T_E(n)\|^4\,dE\,<\,\infty$ almost surely.
Therefore, we get the almost sure a.c. spectrum as before.
\end{proof}

\appendix


\section{Operators with one propagating channel \label{app:spectrum}}

In this appendix we will prove the theorems of Section~\ref{sub:spectrum} and
consider the operator $H$ as in \eqref{eq-H-prop},
$$
(H\psi)(n)\,=\,\phi_n\,\left(\phi_{n+1}^* \psi(n+1)+\phi_{n-1}^*\psi(n-1) \right)\,+\,V_n\phi_n
$$
with Dirichlet boundary conditions $\psi(-1)=0$ where $\psi=\bigoplus_{n=0}^\infty \psi(n)\in\bigoplus_n\CC^{s_n}=\ell^2(\Ab_\bfs)$, 
and $\phi_n\in\CC^{s_n}$. The sequences $\phi_n,\, V_n$ are chosen such that $H$ restricted to $\Dd_0$ as in \eqref{eq-def-D0} is essentially self-adjoint.   
As in \eqref{eq-def-Phi} we let $\Phi_n=P_n\phi_n=\bigoplus_{k=0}^{n-1}\nul\oplus \phi_n \oplus \bigoplus_{k=n+1}^\infty \nul$. 
For the spectral theory recall Proposition~\ref{prop-spec-decom} which states some possibly trivial eigenvalues and eigenvectors of $H$
in the orthogonal complement of $\VV=\bigoplus_n \VV_n$,  the cyclic space generated by all the $\Phi_n$.
We may therefore restrict the investigation of the spectral theory and Green's functions to this space.
Recall from \eqref{eq-def-Tt} that we defined the transfer matrices
\begin{equation*}
 T_{z,n}:=
 \pmat{\left(\phi_n^* (z-V_n)^{-1} \phi_n\right)^{-1} & -1 \\ 1 & 0 } \qtx{and}
 T_z(n)=T_{z,n} T_{z,n-1}\,\cdots\,T_{z,1}T_{z,0}\,.
\end{equation*}
For $z=E$ in the spectrum of $V_n$ we have the holomorphic extension $T_{E,n}=\smat{0&-1\\1&0}$.
As the introduction of the transfer matrices in \eqref{eq-def-Tt} suggests one can use the arsenal of transfer matrix methods developed for one-dimensional Jacobi operators in this setup.

\subsection{Green's function identities}

Let us start with the Green's functions and consider truncated operators with different boundary conditions.
Hence, let $H_{N,\beta}$ denote the operator $H$ restricted to $\bigoplus_{n=0}^N \ell^2(S_n)\cong\bigoplus_{n=0}^N \CC^{s_n}$ with the boundary condition
$\phi_{N+1}^* \psi(N+1)=-\beta \phi_N^* \psi(N)$ at $N$, i.e. for $\beta=0$ we have Dirichlet boundary conditions and
$H_{N,\beta}=H_{N,0}-\beta |\Phi_N\rangle \langle\Phi_N|$.
For $\im z>0$ we define the {\it radial components} of the resolvent by
\begin{equation}\label{eq-def-g}
 g_z(m,n):=\langle \Phi_m\,|\, (H-z)^{-1}\,|\, \Phi_n\rangle \;,\quad
 g_{z,N,\beta}(m,n):=\langle \Phi_m\,|\,(H_{N,\beta}-z)^{-1}\,|\,\Phi_n\rangle\;.
\end{equation}
Clearly, $H_{N,\beta}\psi \to H \psi$ for any $\psi\in\Dd_0$ and hence, $H_{N,\beta}\to H$ in strong
resolvent sense. Therefore, $g_{N,\beta}(m,n)\to g(m,n)$ for $N\to\infty$ and varying $\beta$.

Let $u_{z,n}$ and $v_{z,n}$ be solutions of 
the modified eigenvalue equation \eqref{eq-eig-u} with $u_{z,-1}=0=v_{z,0}$ and $u_{z,0}=v_{z,-1}=1$, then
\begin{equation}\label{eq-def-u-v}
T_z(n)=\pmat{u_{z,n} & v_{z,n} \\ u_{z,n-1} & v_{z,n-1}} \;.
\end{equation}
Moreover, let  $w_z^{(N,\beta)}$ be a solution with $\beta$-boundary condition at $N$, i.e.
\begin{equation}\label{eq-def-w-beta}
w^{(N,\beta)}_{z,N}=1\,,\quad w^{(N,\beta)}_{z,N+1}=-\beta\;,\quad \pmat{w^{(N,\beta)}_{z,n+1}\\w^{(N,\beta)}_{z,n}}=
T_{z,n} \pmat{w^{(N,\beta)}_{z,n}\\w^{(N,\beta)}_{z,n-1}}\,.
\end{equation}
By self-adjointness and hence existence of $(H-z)^{-1}$, for $\im(z)>0$ there exists a unique solution $w_z^{(\infty)}$ such that
$$
w^{(\infty)}_{z,-1}=-1\;,\quad \sum_{n=0}^\infty |w^{(\infty)}_{z,n}|^2\,\|\psi_{E,n}\|^2
\,<\,\infty\;.
$$
In particular one finds for $\psi=\bigoplus_{n\geq 0} w^{(\infty)}_{z,n} \psi_{E,n}$ that $(H-z)\psi=\Phi_0$.
For two solutions $u ,v$ of \eqref{eq-eig-u} we further define the Wronskian
\begin{equation}
 W(u,v)=u_{n+1} v_n - u_n v_{n+1}
\end{equation}
which is independent of $n$.
\begin{lemma}\label{lem:gr}
For $\im(z)>0$ we find the following identities:
\begin{enumerate}[{\rm (i)}]
\item Any solution $w_z$ (in particular, $u_z, v_z, w^{(N,\beta)}_z, w^{(\infty)}_z$) of the modified eigenvalue equation \eqref{eq-eig-u} satisfies
\begin{equation} \label{eq-sum-|u|^2}
\im(z)\sum_{n=0}^N \frac{|w_{z,n}|^2}{\|\phi_n\|^2}\leq
\im(z) \sum_{n=0}^N |w_{z,n}|^2\,\|\psi_{z,n}\|^2\,=\,
\im(w_{z,N+1}\,\overline{w_{z,N}}-w_{z,0}\,\overline{w_{z,-1}})\,.
\end{equation}
\item We have
\begin{equation}\label{eq-gr-beta}
 g_{z,N,\beta}(m,n) = \begin{cases}[W(w_z^{(N,\beta)},u_z)]^{-1} w^{(N,\beta)}_{z,m} u_{z,n}& \quad \text{for $N\geq m\geq n \geq 0$}  \\ 
 [W(w_z^{(N,\beta)},u_z)]^{-1} u_{z,m} w^{(N,\beta)}_{z,n} & \quad \text{for $N\geq n\geq m \geq 0$} \end{cases}
\end{equation}
and similarly
\begin{equation}\label{eq-gr}
 g_{z}(m,n) = \begin{cases} w^{(\infty)}_{z,m} u_{z,n}& \quad \text{for $m\geq n \geq 0$}  \\ 
 u_{z,m} w^{(\infty)}_{z,n} & \quad \text{for $n\geq m \geq 0$} \end{cases}
\end{equation}
For the last equation, note that $W(w^{(\infty)}_z,u_z)=w^{(\infty)}_{z,0}u_{z,-1}-w^{(\infty)}_{z,-1}u_{z,0}=1$.
\item Let $P_n:\CC^{s_n}\to \ell^2(S_n)\subset \ell^2(\Ab_\bfs)$ be the canonical injection so that
$P_n\phi_n=\Phi_n$. Then $P_n^*$ is the canonical projection from $\ell^2(\Ab_\bfs)$ to $\CC^{s_n}\cong \ell^2(S_n)$ and for $m\neq n$ one finds
\begin{equation}\label{eq-resolvent}
 P_m^*(H-z)^{-1} P_n = g_z(m,n)\;\frac{(V_m-z)^{-1} \phi_m \phi_n^*(V_n-z)^{-1}}{(\phi_m^*(V_m-z)^{-1}\phi_m)(\phi_n^*(V_n-z)^{-1}\phi_n)}
 =g_z(m,n) \psi_{z,m} \psi_{\bar z,n}^*
\end{equation}
and for $m=n$,
\begin{align}\notag
 P_n^*(H-z)^{-1}P_n &= g_z(n,n)\;
 \psi_{z,n} \psi_{\bar z,n}^*\,
 +(V_n-z)^{-1}\left(\one-\phi_n\,\psi_{\bar z,n}^*\right)\\
&= \left(V_n-z+\frac{\phi_n\phi_n^*}{g_z(n,n)}-\frac{\phi_n\phi_n^*}{\phi_n^*(V_n-z)^{-1}\phi_n} \right)^{-1}\;.
\label{eq-resolvent2}
 \end{align}
Changing $H$ with $H_{N,\beta}$ one has the same formulas with $g_z$ changed to $g_{z,N,\beta}$.
\end{enumerate}
\end{lemma}
Note that the $\phi_n,\,\psi_{z,n}$ are column vectors, hence expressions like $\phi_m \phi_n^*$ or $\psi_{z,m} \psi_{\bar z,n}^*$ 
are $s_m \times s_n$ matrices that may also be written as $|\psi_{z,m}\rangle \langle \psi_{\bar z,n}|$, but we only want to use the Dirac notation
for operators on the complete Hilbert space $\ell^2(\Ab_\bfs)$. Also note that
$\psi_{\bar z,n}^* = a_{z,n} \phi_n^*\,(z-V_n)^{-1} $ as the complex conjugation included in the adjoint changes $\bar z$ back to $z$.

\begin{proof}
For (i) note that \eqref{eq-eig-u} and \eqref{eq-def-a-psi} imply
$$
w_{z,n+1} \overline{w_{z,n}} + w_{z,n-1} \overline{w_{z,n}}=(\phi_n^*(z-V)^{-1}\phi_n)^{-1}\, |w_{z,n}|^2 \,=\,a_{z,n}\,|w_{z,n}|^2\,.
$$
Taking imaginary parts, summing over $n$ and using \eqref{eq-rel-abpsi} yields \eqref{eq-sum-|u|^2}.

\vspace{.2cm}

For parts (ii) and (iii), let $\psi$ be the solution of $(H_{N,\beta}-z)\psi = \Phi$ or $(H-z)\psi=\Phi$ respectively,
where $\Phi=P_n\varphi$, i.e. $\Phi(l)=\delta_{n,l}\varphi$, $\varphi\in\CC^{s_n}$.
Then, for $m\geq n$ one must have
 $\Phi_m^*\psi=c_1 w^{(N,\beta)}_{z,m}$, or $c_1 w^{(\infty)}_{z,m}$, respectively, and for $m\leq n$ one has
 $\Phi_m^*\psi=c_2u_{z,m}$. Moreover, for $m\neq n$,
 $$
 \psi(m)=(z-V_m)^{-1} \phi_m (\langle \Phi_{m+1}|\psi\rangle\,+\,\langle\Phi_{m-1}|\psi\rangle)
=\psi_{z,m} \,\langle\Phi_{m}|\psi\rangle\,,
 $$ 
where you should note that $\langle\Phi_m|\psi\rangle=\phi_m^*\psi(m)$.
Furthermore,
 \begin{align*}
\psi(n) \,&=\,(z-V_n)^{-1}\big( \phi_n(\langle \Phi_{n+1} | \psi\rangle\,+\,\langle\Phi_{n-1}| \psi\rangle\,)- \varphi\big)\\
&=\,(V_n-z)^{-1}\left(\phi_n\frac{\langle \Phi_n | \psi\rangle }{\phi_n^*(V_n-z)^{-1}\phi_n}-\phi_n \,\psi_{\bar z,n}^*\varphi+\varphi \right),
 \end{align*}
 note that $\psi_{\bar z,n}^*\varphi=(\phi_n^*(z-V_n)^{-1}\varphi)\, /\, (\phi_n^*(z-V_n)^{-1}\phi_n)$. 
 This implies in the case of the operator $H_{N,\beta}$ that
 $$
 \pmat{c_1w^{(N,\beta)}_{z,n+1} \\ c_1 w^{(N,\beta)}_{z,n}}=\pmat{\langle \Phi_{n+1}|\psi\rangle \\ \langle\Phi_n| \psi \rangle} = 
 T_{n,z} \pmat{\langle \Phi_{n} | \psi \rangle \\ \langle \Phi_{n-1} | \psi \rangle} + 
 \pmat{\psi_{\bar z,n}^*\varphi\\0} = \pmat{c_2 u_{z,n+1}+a \\ c_2 u_{z,n}}\;.
 $$
Some algebra then gives
$$
c_1= \psi_{\bar z,n}^*\varphi\,u_{z,n} / W(w^{(N,\beta)}_z,u_z) \qtx{and} c_2 = \psi_{\bar z,n}^*\varphi\,w^{(N,\beta)}_{z,n} / W(w^{(N,\beta)}_z,u_z)\;.
$$
The analogue equations hold when replacing $H_{N,\beta}$ with $H$.
Noting that $\psi_{\bar z,n}^*\phi_n=1$ the case $\varphi=\phi_n$ gives \eqref{eq-gr-beta} and \eqref{eq-gr}.  
Together with the expressions for $\psi(m),\, \psi(n)$ above, \eqref{eq-resolvent} and \eqref{eq-resolvent2} follow.
\end{proof}

\subsection{Proof of Theorem~\ref{th:sigma_0} \label{sec:sigma_0}}

An immediate consequence of the above calculations is the following lemma which also proves parts (i) and (ii) of Theorem~\ref{th:sigma_0}.
As the sets $A_n$ and $B_n$ in Theorem~\ref{th:sigma_0} are finite, it will be enough to consider compact intervals inside the complements.
\begin{lemma}\label{lem:msr}
 We let $\mu_n$ denote the spectral measure at $\Phi_n$, i.e. $\int f\, d\mu_n = \langle\Phi_n|f(H)|\Phi_n\rangle$. 
\begin{enumerate}[{\rm (i)}]
\item Assume that for all $E\in[a,b]$ and $k=0,1,\ldots,n-1$ the transfer matrices $T_{E,k}$ exist
in the sense as in \eqref{eq-def-Tt} and \eqref{eq-def-T_E}, i.e. $[a,b]\cap B_n=\emptyset$.
 Then restricted to the interval $[a,b]$ one finds
 $$
 1_{[a,b]}(E)\;\mu_n(dE)\, =\, 1_{[a,b]}(E)\,u_{E,n}^2\, \mu_0(dE)\,.
 $$
 In particular, in the interval $[a,b]$ the measure $\mu_n$ is continuous with respect to $\mu_0$.
\item For $\varphi\in \VV_n$ let $\mu_{n,\varphi}$ denote the spectral measure at $P_n\varphi$, i.e.
$\int f d\mu_{n,\theta}=\langle P_n\varphi |f(H)| P_n\varphi\rangle$. Assume that for all $E\in[a,b]$ the matrix $T_{E,n}$ exists 
(as expressed above). Then, one finds
$$
 1_{[a,b]}(E)\;\mu_{n,\varphi}(dE)\, =\, 1_{[a,b]}(E)\,\left|\varphi^*\psi_{E,n}\right|^2 \, \mu_n(dE)\,
$$
where $\varphi^*\psi_{E,n}$ has only finitely many zeros in $[a,b]$ for $\varphi\neq 0$.
 \end{enumerate}
 \end{lemma}
\begin{proof}
 By Lemma~\ref{lem:gr} one has $\mu_n(dE)=\lim_{\eta\to 0} \frac1\pi \im(w^{(\infty)}_{E+i\eta,n} u_{E+i\eta,n})\,dE$ where the limit has to be understood in the vague sense.
 For the case $n=0$ note that $u_{z,0}=1$ for all $z$.
Now, by the definitions one easily sees $w^{(\infty)}_{z,n}=w^{(\infty)}_{z,0} u_{z,n}-v_{z,n}$ and hence
$w^{(\infty)}_{z,n} u_{z,n}= w^{(\infty)}_{z,0} u_{z,n}^2-v_{z,n}u_{z,n}$.
Thus we find 
$$|\im(w^{(\infty)}_{z,n} u_{z,n})-\re(u_{z,n}^2)\im(w^{(\infty)}_{z,0})|\leq 
|\re(w^{(\infty)}_{z,0})|\,|\im (u_{z,n}^2)|+|\im(u_{z,n}v_{z,n})|\;.$$
Using $w^{(\infty)}_{z,0}=\int (E'-z)^{-1} d\mu_0(dE)$, and splitting the integral one finds for any $\varepsilon$ that
$|\re(w^{(\infty)}_{z,0}|\leq 1/\varepsilon \im(w^{(\infty)}_{z,0})+\varepsilon/\eta$.
Moreover, by the assumption, $u_{E,n}$ and $v_{E,n}$ are well defined for $E\in[a,b]$ and holomorphic in $E$,
therefore $|\im (u_{z,n}^2)|,|\im(u_{z,n}v_{z,n})|<C\eta$ for small imaginary part $\eta$  uniformly in $[a,b]$.
Putting these estimates together, one finds $\eta_\varepsilon$ for any $\varepsilon$ such that
$$
\big|\im(w^{(\infty)}_{E+i\eta,n} u_{E+i\eta,n})-u_{E,n}^2 \im(w^{(\infty)}_{E+i\eta,0})\big|\,<\,\varepsilon+ \varepsilon\,\im(w^{(\infty)}_{E+i\eta})
$$
for any $0<\eta<\eta_\varepsilon$ and any $E\in[a,b]$.
Therefore, for any $f\in C([a,b])$ one finds
$$
\left|\int_a^b f(E) \mu_n(dE) - \int_a^b f(E) u_{E,n}^2 \mu_0(dE) \right|< \int_a^b |f(E)|\,(\varepsilon\, dE\,+\,\varepsilon \mu_{0}(dE))
$$
which goes to zero for $\varepsilon\to 0$ as $\mu_0$ is a bounded measure. This finishes part (i). 

For part (ii) let us first see that $\psi_{z,n}$ indeed extends holomorphically to $[a,b]$.
The only critical values are the eigenvalues of $V_n$ restricted to $\VV_n$.
Thus, let the eigenvalue decomposition be given by $V_n|\VV_n = \sum_{k} e_k v_k v_k^* $ where the $v_j$ form an orthonormal basis
of $\VV_n$. As $\phi_n$ is a cyclic vector, all the eigenvalues $e_j$ are different.
Then
$$
\psi_{z,n} = \frac{\sum_j\alpha_j \alpha_j^* \phi_n /(e_j-z)}{\sum_j |\alpha_j^*\phi_n|^2/(e_j-z)}
\qtx{giving the extension}\psi_{e_k,n}= \alpha_k/ (\phi_n^*\alpha_k)\,.
$$
As $\phi_n$ is a cyclic vector for $V_n|\VV_n$, $\phi_n^*\alpha_k\neq 0$ and $\psi_{z,n}$ is holomorphic at $z=e_k$.
From \eqref{eq-resolvent2} we find
\begin{align*}
&\langle P_n\varphi|(H-z)^{-1}|P_n \varphi \rangle\,=\, \\
& \qquad \qquad g_z(n,n) \varphi^*\psi_{z,n} \psi_{\bar z,n}^* \varphi+
\varphi^*(V_n-z)^{-1}\varphi-(\varphi^*\psi_{z,n}\, \psi_{\bar z,n}^* \varphi) \,(\phi_n^*(V_n-z)^{-1}\phi_n)
\end{align*}
Using $\psi_{e_k+z,n}=\alpha_k/(\phi_n^*\alpha_k)+\Oo(z)$ and the spectral decomposition of $V_n|\VV_n$ as above one easily checks that 
the sum of the last two terms extends holomorphically to $z=e_k$ and
hence defines an analytic function for $z=E\in[a,b]$ with zero imaginary part.
Moreover, $f(z)=\varphi^*\psi_{z,n} \psi_{\bar z,n}^* \varphi$ is holomorphic for $z\not\in\RR$ and
for $z=E\in[a,b]$. By cyclicity of $\phi_n$, $f(z)$ is not the zero function, and $f(E)=|\varphi^*\psi_{E,n}|^2$ has only finitely many zeros 
in $[a,b]$. A similar argument as in (i) now gives $\mu_{n,\varphi}(dE)=f(E)\mu_n(dE)$ for the measures restricted to $[a,b]$.
\end{proof}

In order to prove Theorem~\ref{th:sigma_0}~(iii) (part (i) and (ii) are Corollaries of Lemma~\ref{lem:msr})
we will consider an average over the boundary conditions $\beta$ for the finite matrices $H_{N,\beta}$.
By the definition of $w^{(n,\beta)}_z$ as in \eqref{eq-def-w-beta} one finds
$$
T_z(n+1)\smat{w^{(n,\beta)}_{z,0}\\w^{(n,\beta)}_{z,-1}}\,=\,\pmat{-\beta\\1}\;.
$$ 
Using this equation together with \eqref{eq-def-u-v}, \eqref{eq-gr-beta}, the fact that the transfer matrices have determinant one, 
as well as $W(w_z^{(n,\beta)},u_z)=-w^{(n,\beta)}_{z,-1}$, one obtains
$$
m_{n,\beta}(z):=g_{z,n,\beta}(0,0)=\langle \Phi_0\,|\,(H_{n,\beta}-z)^{-1} \,|\,\Phi_0\rangle=
\frac{\beta v_{z,n}+v_{z,n+1}}{\beta u_{z,n}+u_{z,n+1}}\;.
$$
The first equation defines $m_{n,\beta}(z)$, the second one reminds of the definition of $g_{z,n,\beta}$ in \eqref{eq-def-g}.
For fixed $z$ with $\im(z)>0$ and $\beta$ varying along $\RR$, $m_{n,\beta}(z)$ 
forms a circle, the Weyl circle, in the upper half plane. Moreover, by \eqref{eq-sum-|u|^2} we have 
$\im(u_{z,n+1}\overline{u_{z,n}})>0$ and $\im(\beta)>0$ implies
$\im(\beta|u_{n,z}|^2+u_{z,n+1}\overline{u_{z,n}})>0$ showing that the denominator will not be $0$. 
Hence, $m_{n,\beta}(z)$ can be extended to $\im(\beta)\geq 0$ and for $\im \beta>0$,
$m_{n,\beta}(z)$ lies inside the Weyl circle.
Now, setting $\beta=i$, $m_{n,i}(z)$ is a Herglotz function in $z$ for $\im(z)>0$. The corresponding measure, $\mu_{n,i}$, 
equals the integration of the measures $\mu_{n,\beta}$  associated to $m_{n,\beta}(z)$ over the Cauchy distribution in $\beta$.
The measure $\mu_{n,i}$ is also given by the distributional limit of
$\lim_{\eta\to0} \frac1\pi \im(m_{n,i}(E+i\eta))dE$. We find
$$
\lim_{\eta\to 0} \im(m_{n,i}(E+i\eta))\,=\,
 \im\left(\frac{i v_{E,n}+v_{E,n+1}}{i u_{E,n}+u_{E,n+1}}\right) = \frac{1}{u_{E,n}^2+u_{E,n+1}^2}
$$
for energies $E\not\in B_{n+1}$ where all $T_{E,m}$ for $m\leq n$ exist.
Here we used that the Wronskian is one, $W(u_E,v_E)=1$.

This measure is absolutely continuous except for possibly some points $E'$ in the set $B_{n+1}$.
It is clear from rank one perturbation theory that these delta measures in $\mu_{n,i}$ can only come from energies that are
eigenvalues of $H_{n,\beta}$ for all $\beta$ with eigenvectors that are fixed in $\beta$ and orthogonal to $\Phi_n$.
As shown in Remark~\ref{rem:purepoint}~(iii) and (iv) one can indeed construct such eigenvectors under certain conditions.
Thus, we may define $\nu_n$ to be the pure point part of $\mu_{n,i}$ (and in fact of all $\mu_{n,\beta}$) supported on $B_{n+1}$. 
As these eigenvectors will remain compactly supported eigenvectors for $n\to\infty$ and $n=\infty$, 
the sequence $\nu_n$ is increasing and has a limit $\nu$ supported on $B_\infty$. Remark~\ref{rem:purepoint}~(iii) also clasifies when $\nu(\{E\})>0$.

As $H_{n,\beta_n}\to H$ in strong resolvent sense, there must be a unique limit point $\lim_{n\to\infty} m_{n,\beta_n}(z)=\langle \Phi_0|(H-z)^{-1}|\Phi_0\rangle=w^{(\infty)}_{z,0}$ for all 
$\im \beta\geq 0$.
Hence, one also has $\lim_{n\to\infty} m_{n,i}(z)=w^{(\infty)}_{z,0}$ and therefore it follows in the weak sense that $\mu_{n,i}\to \mu_0$ and by the considerations above,
$$
\mu_0(dE)\,=\,\lim_{n\to\infty} \left( \frac{1}{\pi}\;\frac{1_{\RR\setminus B_\infty}(E)\;dE}{|u_{E,n}|^2+|u_{E,n+1}|^2}\,+\,\nu_n(dE)\right)\,=\,
\lim_{n\to\infty} \frac{1_{\RR\setminus B_\infty}(E)\;dE}{\pi\,\|T_E(n+1)\smat{1\\0}\|^2}\,+\,\nu(dE)\,.
$$
This completes the proof of Theorem~\ref{th:sigma_0}. \qed

\begin{remark}\label{rem:selfad}
 The radius $r_n(z)$ of the Weyl circle is given by
 $$
 2r_n(z)\,=\,\sup_{\beta\in\RR}\left|\frac{\beta v_{z,n}+v_{z,n+1}}{\beta u_{z,n}+u_{z,n+1}}-\frac{v_{z,n}}{u_{z,n}} \right|\,=\,\sup_{\beta\in\RR} 
 \frac{1}{|u_{z,n}|}\;\frac{1}{|\beta u_{z,n}+u_{z,n+1}|}\,\leq\,\frac{1}{\im(\overline{u_{z,n}} u_{z,n+1})}
 $$
 where we used $W(u,v)=1$ and $|u_{z,n}|=|\overline{u_{z,n}}|$. Using $u_{z,0}=0$ and \eqref{eq-sum-|u|^2} we get 
 $$
 2r_n(z)\,\im(z)\,\leq\, \left(\sum_{k=0}^n \frac{|u_{z,k}|^2}{\|\phi_n\|^2}\right)^{-1}\,\leq\,\left(\sum_{k=0}^{n-1} \frac{|u_{z,k}u_{z,k+1}|}{\|\phi_k\|\,\|\phi_{k-1}\|}\right)^{-1}
 \,\leq\, \left(\sum_{k=0}^{n-1} \frac{C}{\|\phi_k\|\,\|\phi_{k-1}\|}\right)^{-1}\,.
 $$
 Therefore, if $\sum_n \|\phi_n \phi_{n-1}\|^{-1}=\infty$, then $r_n(z)\to 0$ for $\im(z)>0$ (limit point case), 
 $H_{n,\beta}\to H$ in strong resolvent sense and the compactly supported vectors $\Dd_0$ form a core.
\end{remark}




\subsection{Subordinacy theory \label{sub:subordinacy} }

Analogue to above let us now define the $m$-function for the infinite operator $H$ by
$$
m(z):=g_z(0,0)\qtx{implying} w_{z,n}^{(\infty)}\,=\,m(z)\,u_{z,n}\,-\,v_{z,n}
$$
or short $w^{(\infty)}_z=m(z)u_z-v_z$. 
Note that $u_z$ and $v_z$ play the role of $\psi(z)$ and $-\varphi(z)$ as in \cite{KP}.

An essential support of the a.c. part of the measure $\mu_0$ (spectral measure at $\Phi_0$), and hence of the a.c. spectrum of $H$, is given by the set of energies
$$
\Sigma'_{ac}\,:=\,\{E\in \RR\,:\, m(E):=\lim_{\eta\downarrow 0} m(E+i\eta)\;\;\text{exists and}\;\; \im m(E)\,>\,0\}\;.
$$
Similarly, an essential support of the singular part of $\mu_0$ is given by
$$
\Sigma'_{s}\,:=\,\{E\in\RR\,:\,\lim_{\eta\downarrow 0}\,\im\,m(E+i\eta)\,=\,\infty\,\}\;.
$$
Recall that we call a non-zero solution $w=(w_{n})_n$ of \eqref{eq-eig-u} at energy $E$ subordinate if for any linear independent solution  
$\wt w$ one finds $\lim_{n\to \infty} \|w\|_{E,n}\,/\,\|\wt w\|_{E,n}\,=\,0 $ where we define
$$
\|w\|^2_{z,n}\,:=\,\sum_{k=0}^n |w_k|^2\,\|\psi_{z,k}\|^2
$$
for any complex energy $z\in \CC \setminus \overline{B_\infty}$ and sequence $w_n$.
Following Kahn-Pearson one can show the following analogue to \cite[Theorem~1 and Theorem~2]{KP}.
\begin{theorem} \label{th:subord-app} Let $E \in \RR \setminus \overline{B_\infty}$, i.e. all transfer matrices exist in a neighborhood of $E\in \RR$.
 \begin{enumerate}[{\rm (i)}]
  \item If $m(E)$ exists and is real, then $m(E) u_E-v_E$ is subordinate. 
  \item If $\lim_{\eta\to 0} |m(E+i\eta)|=\infty$ then $u_E$ is subordinate.
  \item If $w^{(\infty)}_E:=m u_E-v_E$ is subordinate, then $m$ is real and along some sequence $\eta_j\to 0$ we find 
  $$\lim_{j\to\infty}\,m(E+\eta_j)\,=\,m \qtx{and hence} \lim_{j\to \infty} w^{(\infty)}_{E+\eta_j,n}=w^{(\infty)}_{E,n}\;.$$
  \item If $u_E$ is subordinate, then $\lim_{j\to\infty} |m(E+i\eta_j)|\,=\,\infty$\;.
 \end{enumerate}
\end{theorem}
Note that by general theory about Herglotz functions the limit $m(E)$ does exist for Lebesgue almost every $E$, hence the limits along sequences $\eta_j$ are limits $\eta\downarrow 0$ for
Lebesgue almost all $E$ where one has a subordinate solution.
Using the fact that $m(z)=\int (E-z)^{-1} \mu_0(dE)$, Theorem~\ref{th:subordinacy} immediately follows by standard arguments as in \cite{KP}.

For the proof we focus on (i) and (iii), parts (ii) and (iv) follow similarly by 
considering $ w^{(\infty)}_z/m(z)=u_z-v_z/m(z)$, which replaces the role of $m(z)$ by $1/m(z)$ and reverses the role of $u_z$ and $v_z$.
As in \cite{KP} the following estimate is important.
Note that by \eqref{eq-sum-|u|^2} and the fact that $w_{z,n}^{(\infty)} \to 0$ for $n\to\infty$ one finds
\begin{equation}\label{eq-basic-est}
\|u_z\,m(z)\,-\,v(z)\|_{z,n}\,=\,\sum_{k=0}^n |w_{z,k}^{(\infty)}|^2 \,\|\psi_{z,k}\|^2\,
\leq\,\frac{\im(m(z))}{\im(z)}\;.
\end{equation}

This estimate and Lemma~\ref{lem-sub-keyest} are crucial to make the proof of \cite[Theorem~1]{KP} work.
So let $E\in \RR \setminus \overline{B_\infty}$ and
for $\eta>0$ define the operator $L=L_\eta$ acting on a sequence $w=(w_n)_{n\geq 0}$ by
\begin{equation}
(L w)_n\,:=\,-v_{E,n} \sum_{k=0}^n \left( u_{E,k} (a_{E+i\eta,k}-a_{E,k}) w_k \right)\,+\,u_{E,n} \sum_{k=0}^n \left( v_{E,k} (a_{E+i\eta,k}-a_{E,k})\, w_k \right)\,.
\end{equation}
Then it is straight forward to verify that
$ w^{(\infty)}_z=m(z)u_z-v_z=m(z)u_E-v_E+Lw^{(\infty)}_z$ with $z=E+i\eta$.
Using Lemma~\ref{lem-sub-keyest}~(i) and Cauchy-Schwarz we find 
$$\left| \sum_{k=0}^n v_{E,k} \|\psi_{E,k}\| \frac{(a_{E+i\eta,k}-a_{E,k})}{\|\psi_{E,k}\|^2}\, w_k \|\psi_{E,k}\|\right|\,\leq\,
\eta\,\|v_E\|_{E,n}\, \|w\|_{E,n}$$
and (cf. \cite[eq. (27), (28)]{KP}
\begin{equation}\label{eq-est-L}
 \|Lw\|_{E,n}\,\leq\,2\eta\,\|u_E\|_{E,n}\,\|v_E\|_{E,n}\,\|w\|_{E,n} \qtx{and hence} \|L\|_{E,n}\,\leq\,2\eta\,\|u_E\|_{E,n}\,\|v_E\|_{E,n}\;.
\end{equation}
Assume now $m(E+i\eta)\to m(E)\in \RR$ for $\eta\downarrow 0$ and as in Lemma~3~(i) of \cite{KP} define $\eta_n$ to be the smallest positive number such that
$$
\eta_n\,=\,\sqrt{\im(m(E+i\eta_n))}\,/\,\left({\sqrt{\|u_E\|_{E,n} \|v_E \|_{E,n}}\,\big[\|u_E\|_{E,n}+\|v_E\|_{E,n}\big]}\right)\,.
$$
Then $\eta_n\to 0$ and for $z_n=E+i\eta_n$ we find as in \cite[Theorem~1]{KP} (cf. \cite[eq. (29)]{KP}
$$
\lim_{n\to\infty} \frac{\|u_{z_n}\, m(z_n)\,-\,v_{z_n}\|_{E,n}}{\|u_E\,m(z_n)\,-\,v_E\|_{E,n}}\,=\,1\,,\quad
\lim_{n\to\infty} \left(\frac{\im(m(z_n))}{\eta_n}\right)^{1/2} \frac{1}{\|u_E\|_{E,n}+\|v_E\|_{E,n}}\,=\,0\;.
$$
The latter estimate corresponds to the term $F_3$ in \cite{KP}. The term corresponding to $F_2$, however, needs a slight modification here at first sight, which is that we have to use the norm 
$\|\cdot\|_{z_n,n}$ in the numerator. More precisely, combining the above estimates with \eqref{eq-basic-est} now gives
$\|w^{(\infty)}_{z_n,n}\|_{z_n,n}\,/\,\left(\|u_E\|_{E,n}+\|v_E\|_{E,n} \right)\,\to\,0\;$.
It is precisely at this point that Lemma~\ref{lem-sub-keyest}~(ii) is crucial to change the $\|\cdot\|_{z_n,n}$ norm to the $\|\cdot\|_{E,n}$ norm and to obtain
\begin{equation}
 \lim_{n\to \infty} \frac{\|u_{z_n}\,m(z_n)\,-\,v_{z_n}\|_{E,n}}{\|u_E\|_{E,n}+\|v_E\|_{E,n}}\,=\,0\qtx{and}
 \lim_{n\to\infty} \frac{\|u_E\,m(E)\,-\,v_E\|_{E,n}}{\|u_E\|_{E,n}+\|v_E\|_{E,n}}\,=\,0\;,
\end{equation}
giving the subordinacy of $u_E\, m(E) - v_E$. This proves part (i) of Theorem~\ref{th:subord-app}.

For part (iii) let $u_E\,m-v_E$ be a subordinate solution, then clearly, $m\in\RR$.
We can basically follow the proof of \cite[Theorem~2]{KP}. In order to get to the equivalent of \cite[equation (33)]{KP} we need to use 
both estimates of Lemma~\ref{lem-sub-keyest} again (part (i) for the bound on the operator $L$ similar to above and part (ii) to replace the  $\|\cdot\|_{z,n}$ norm by
the $\|\cdot\|_{E,n}$ norm in \eqref{eq-basic-est}) and obtain
$$
\lim_{n\to\infty} \frac{\|u_E\,m(z_n)\,-\,v_E \|_{E,n}}{\|u\|_{E,n}\,(1+\im(m(z_n))}\,=\,0\,
$$
where $z_n=E+i\eta_n$ with 
$$
\eta_n\,=\,\frac{\sqrt{\im(m(E+i\eta_n))}}{\|u_E\|_{E,n}^{3/2}\,\sqrt{1+ \im\,m(E+i\eta_n)} \,\big[\|u_E\|_{E,n}+\|v_E\|_{E,n}\big]}\;.
$$
Using the subordinacy of $u_E\,m-v_E$ we then obtain $m(z_n)\to m$ as in \cite{KP}. 

As mentioned above, parts (ii) and (iv) follow analogously to (i) and (iii), respectively.


\subsection{Proof of Theorem~\ref{th:ac-necessary}}

Theorem~\ref{th:ac-necessary} now follows from the subordinacy theory and Theorem~\ref{th:sigma_0}~(ii).

\begin{proof}[Proof of Theorem~\ref{th:ac-necessary}]
From Theorem~\ref{th:sigma_0}~(ii) we obtain that
$$
\int 1_{\RR\setminus B_\infty}(E)\, u_{E,n}^2\;\mu_0(dE)\,\leq\, \|\Phi_n\|^2\,=\,\|\phi_n\|^2\;.
$$
Following the arguments in \cite{LaSi} we also look at the 'Neumann' boundary conditions at $n=0$.
Thus, let $H^{(1+)}$ denote the operator $H$ restricted to the $\bigoplus_{n\geq 1} \ell^2(S_n)$ with Dirichlet boundary conditions at $n=1$.
Moreover, let $\mu_1^{(1+)}$ denote the spectral measure at $\Phi_1$ of $H^{(1+)}$.
As $v_{E,0}=0,\,v_{E,1}=-1$ we obtain completely analogously that
$
\int 1_{\RR\setminus B_\infty} v_{E,n}^2\;\mu_1^{(1+)}(dE)\,\leq\, \|\phi_n\|^2\;.
$
The standard Green's function recursion in spherical coordinates (cf. \cite{FHS0}) coming from the resolvent identity gives in this case for $g'_{z}:=\langle \Phi_1|(H^{(1+)}-z)^{-1}|\Phi_1\rangle$ that
$g_{z}(0,0)=-\left(g'_{z}+(\phi_0^*(z-V_0)^{-1} \phi_0)^{-1} \right)^{-1}$.
This relation can also be obtained using \eqref{eq-gr} for $H$ and $H^{(1+)}$.
It shows that the singular parts of $\mu_0$ and $\mu_1^{(1+)}$ are mutually singular whereas the
a.c. spectrum has the same support. Now we follow the proof of Proposition~3.3 of \cite{LaSi}, defining the absolutely continuous measure $\mu_{ac}:=\min(\mu_0, \mu_1^{(1+)})$ given by
$\mu_{ac}(S)\,:=\,\inf_{A,B;\,S\subset A\cup B} \left(\mu_0(A)\,+\,\mu_1^{(1+)}(B)\,\right)\;.$
Then the above inequalities show  that
$
\int\frac1n \sum_{k=1}^n \left\|\, \boldsymbol{\Phi}_k^{-1}\, T_E(k) \,\right\|^2\,\mu_{ac}(dE)<4$.
  Fatou's lemma then implies Theorem~\ref{th:ac-necessary}~(i). Furthermore,
$$
\|\,\boldsymbol{\Phi}_k^{-1} \,T_E(k)\,T_E(m)^{-1}\,\boldsymbol{\Phi}_m\,/\,\det(\boldsymbol{\Phi}_m)\,\|\,\leq\,
\|\,\boldsymbol{\Phi}^{-1}_k \,T_E(k) \,\|\;\|\,\boldsymbol{\Phi}_m^{-1}\,T_E(m)\,\|
$$
which together with the Cauchy Schwarz inequality $\int |fg|\,\mu_{ac}\leq\left(\int |f|^2\mu_{ac} \right)^{\frac12}
\left(\int |g|^2\mu_{ac} \right)^{\frac12}$
shows uniform boundedness of the integral of the left hand side over $\mu_{ac}(dE)$. Again, Fatou's lemma then yields part (iii).

So it is only left to show part (ii) analogue as in \cite[Section 3]{LaSi}.
As above let $\vec{u}_\theta=\smat{\cos(\theta)\\\sin(\theta)}$ and define 
$$\vec{u}_\theta(n)\,=\,\bpsi_{E,n}\, T_E(n)\, \vec{u}_\theta \qtx{and}
\vec{v}_\theta(n)\,=\,\bpsi_{E,n}\, T_E(n)\,\vec{u}_{\theta+\frac\pi2}\,.$$
Then, clearly $\sum_{k=1}^n\| \vec{v}_\theta(k)\|^2 \leq \sum_{k=1}^n \|\bpsi_{E,k}\,T_{E}(k)\|^2$.
Moreover, let $J=\smat{0 & 1 \\ -1 & 0}$ be the symplectic form.
Then $\bpsi_{E,k} J \bpsi_{E,k}=\det (\bpsi_{E,k})J$ and as $T_E(k)$ leaves the symplectic form invariant, one obtains
$\langle \vec{u}_{\theta}(k)\,,\,J\,\vec{v}_{\theta}(k)\rangle=\det \bpsi_{E,k}$. This leads to
$$
\left(\sum_{k=1}^n \det \bpsi_{E,k}\right)^2\,\leq\, \sum_{k=1}^n \| \vec{u}_\theta(k)\|\;\| \vec{v}_\theta(k)\|\,\leq\,
\left(\sum_{k=1}^n \| \vec{u}_\theta(k)\|^2 \right)\,\left(\sum_{k=1}^n \vec{v}_\theta(k)\|^2 \right)
$$
which together with the above estimate gives the following analogue of \cite[Lemma~3.1]{LaSi}
\begin{equation}\label{eq-ineq-sub}
\frac{\sum_{k=1}^n \|\vec{v}_\theta(k)\|^2}{\sum_{k=1}^n \|\vec{u}_\theta(k)\|^2}\,\leq\,
\left( \,\frac{\sum_{k=1}^n \|\bpsi_{E,k}\,T_E(k)\|^2}{\sum_{k=1}^n \det \bpsi_{E,k}}\,\right)^2\;.
\end{equation}
Now, letting $u_{\theta,n}=\smat{1&0} T_E(n) \vec{u}_\theta$ be the solution of \eqref{eq-eig-u} we see that
$\|\vec{u}_\theta(n)\|^2=|u_{\theta,n}|^2 \|\psi_{E,n}\|^2 + |u_{\theta,n-1}|^2 \|\psi_{E,n-1}\|^2$. 
By similar arguments as in \cite[Section~3]{LaSi} we see that
the right hand side of
\eqref{eq-ineq-sub} must go to infinity if a subordinate solution $u_{\theta,n}$ exists in the sense of Definition~\ref{def:subordinate}.
Hence, for $E\in \Sigma_\Psi$ no such solution exists.
\end{proof}


\section{Proof of Theorem~\ref{th-lim-log(T)} \label{app-estimate}}

The proof is pretty much the same as in \cite{KLS}. One slight difference is that the support of the distribution of the random variable $x_n$ is
not shrinking in $n$. Therefore, we use Lemma~\ref{lem-conv} instead.
It will be enough to show the limit for $R_n$ for any starting angle $\theta$.
From \eqref{eq-rec-R} we obtain
\begin{equation}\label{eq-log(R_n)}
\log(R_{n+1})\,=\,\log(R_n)\,+\,f(x_n,\bar\theta_n)\qtx{and hence}\log(R_{n+1})\,=\,\sum_{j=0}^n \,f(x_j,\bar \theta_j)
\end{equation}
where
\begin{align}
f(x,\bar \theta)\,&:=\,\frac12 \,\log\left(1\,+\,x_n\,\sin(2\bar\theta)\,+\,x_n^2\,\sin^2(\bar \theta)\,\right) \notag \\ &=\,
\frac12\,\log\left(\left(1\,+\,\frac12 x_n \sin(2\bar\theta)\right)^2\,+\,x_n^2 \sin^4(\bar \theta) \right). \label{eq-def-f(x,th)}
\end{align}
By uniform boundedness of $x_n$ (cf. \eqref{eq-x-bounds}) the expression inside the logarithm is uniformly bounded and uniformly bounded
away from zero. Using $f^{(m)}$ for $\partial f^m / \partial x^m$, a Taylor expansion gives
$$
f(x_n,\bar\theta_n)\,=\,f^{(1)}(0,\bar\theta_n)\,x_n\,+\,\tfrac12f^{(2)}(0,\bar\theta_n)\, x_n^2\,+\,\tfrac16f^{(3)}(0,\bar\theta_n)\,x_n^3\,+\,
\tfrac1{4!} f^{(4)}(\xi_n,\bar \theta_n)\,x_n^4\;
$$
where $\xi_n$ depends on $x_n$ and $\bar\theta_n$.
By compactness of the support of $x_n$, $|\tfrac1{4!}f^{(4)}(\xi_n,\bar\theta_n)\,x_n^4|\le C x_n^4$ for some uniform constant $C$.

As $s_n\sim n^{-\alpha}$ we have $\EE(x_n^m)=\Oo(n^{-\alpha})$ and by \cite[Lemmas~8.3 and 8.4]{KLS} we can 
replace $x_j^m$ by its expectation $\EE(x_j^m)$ with errors
of order $o\left(\sum_{j=1}^n s_j^{-1} \right)$.
Furthermore, note
$$
\sum_{j=0}^n C \,\EE(x_j^4)\,+\,\left|\sum_{j=0}^n\tfrac16f^{(3)}(0,\bar\theta_j)\,\EE(x_j^3)\,\right|\,\leq\,
\tilde C\,\sum_{j=0}^n s_j^{-2}\,=\,o\left(\sum_{j=1}^n j^{-\alpha} \right)\;.
$$
Thus, up to errors
of order $o\left(\sum_{j=1}^n s_j^{-1} \right)=o\left(\sum_{j=1}^n j^{-\alpha} \right)$ one has
\begin{align*}
 \sum_{j=0}^n f(x_j,\bar\theta_j)\,&=\, \sum_{m=1}^2\sum_{j=0}^n \frac{1}{m!}f^{(m)}(0,\bar\theta_j)\,\EE(x_j^m)
 =\,\sum_{m=1}^2\,\sum_{j=0}^n \frac{1}{m!}\,f^{(m)}(0,\bar\theta_j)\,\frac{h^{2+m}\sigma^2}{\sin^m(k)s_j}\\ 
 &=\,
 \sum_{j=0}^n \frac{h^3\sigma^2}{\sin(k)s_j}\,\sin(2\bar\theta_j) \,+\,\sum_{j=0}^n \frac{h^4\sigma^2}{\sin^2(k)\,s_j}\,\left(\frac18\,-\,\frac14\,\cos(2\bar\theta_j)\,+\,\frac18\,\cos(4\bar\theta_j) \right)\,.
\end{align*}
Therefore, \eqref{eq-log(T)} follows from
\begin{equation}\label{eq-oscil-terms}
\left|\sum_{j=1}^n s_j^{-1}\,\sin(2\bar\theta)\,\right|\,+\,\left|\sum_{j=1}^n s_j^{-1}\,\cos(2\bar\theta)\,\right|\,+\,
\left|\sum_{j=1}^n s_j^{-1}\,\cos(4\bar\theta)\,\right|\,=\,o\left(\sum_{j=1}^n j^{-\alpha} \right)\;.
\end{equation}
Using Lemma~\ref{lem-conv} and the existence of the limit $\lim_{n\to\infty} s_n^{-1} n^\alpha$, this can be proved analogously as in \cite{KLS}. 
To see this, let us recall the following Lemma: 
\begin{lemma}[$\sim$ Lemma 8.5 in \cite{KLS}]\label{lem-ql}
Let $k_0\in\RR$ be not in $\pi \ZZ$. Then there exists integers $q_l\to\infty$ such that for any $\theta_0,\ldots,\theta_{q_l}$,
$$
\left| \sum_{j=1}^{q_l}\,e^{i\theta_j}\right|\,\leq\,1\,+\,\sum_{j=1}^{q_l}\,\left|\theta_j-\theta_0-jk_0 \right|\;.
$$
\end{lemma}
$\phantom{hh} \hfill \Box$

\vspace{.2cm}

For illustration how to obtain \eqref{eq-oscil-terms} let us pick the first 
term.  
Also note that by convergence of $s_n\,/\,n^\alpha$ and $\sum_n s_n^{-1} = \sum_n n^{-\alpha}=\infty$ one obtains
$$
\sum_{j=0}^n s_j^{-1} \,\sin(\bar\theta_j)\,=\, \sum_{j=1}^n c\, j^{-\alpha} \,\sin(\bar\theta_j)\,+\, o\left(\sum_{j=1}^n j^{-\alpha} \right)\;,
$$
thus, we can consider $\sum_j j^{-\alpha} \sin(2\bar\theta_j)$.

So let $q_l$ be the sequence as in Lemma~\ref{lem-ql} for $k_0=2k$.
By Lemma~\ref{lem-conv} we find almost surely a random sequence $n_l>q_l^2$ such that for any $n>n_l$ one has $|x_n|\,<\,q_l^{-2}$ and 
$q_l^{1-\alpha/2} (n+q_l)^{-\alpha} \geq n^{-\alpha}$.
Then one obtains for $n>n_l$ that
$$
\left|\bar \theta_{n+j}\,-\,\bar\theta_n\,-\,jk\right|\,\leq\, \sum_{i=1}^j |x_{n+i}|\,\leq\, j\,q_l^{-2}\;.
$$
Let $N=n_l+K q_l$, then
\begin{align}
\sum_{j=n_l+1}^N j^{-\alpha} \,\sin(2\bar\theta_j)\,&\leq\,
\sum_{m=0}^{K-1}\sum_{j=1}^{q_l}\bigg[ (n_l+mq_l)^{-\alpha}\, 2\,\left|\bar\theta_{n_l+mq_l+j}-\bar\theta_{n_l+mq_l}-jk \right| \,+\, \notag \\
& \qquad\qquad\qquad\qquad\left| (n_l+mq_l+j)^{-\alpha}-(n_l+mq_l)^{-\alpha}\right|\bigg] \notag \\
& \leq\,\sum_{m=0}^{K-1} (3\,+\,\alpha)\,(n_l+mq_l)^{-\alpha} \,\leq\, 4\,q_l^{-\alpha}\,\sum_{j=n_l+1}^N j^{-\alpha}
\;.\notag 
\end{align}
In the second estimate  we used concavity of the function $-x^{-\alpha}$ for the second term giving that tangents always lie above the graph and hence
for $N>q_l^2$,
$$
\sum_{j=1}^{q_l}\,N^{-\alpha}\,-\,(N+j)^{-\alpha}\,\leq\, \sum_{j=1}^{q_l} \alpha N^{-\alpha-1}\,j\,\leq\,\alpha\,N^{-\alpha}\,q_l^2/N\,\leq\,
\alpha\,N^{-\alpha}\;.
$$
The last estimate comes from $q_l^{-\alpha/2} \sum_{j=1}^{q_l} (n+j)^{-\alpha} \geq q_l^{1-\alpha/2} (n+q_l)^{-\alpha}\geq n^{-\alpha}$ for all
$n>n_l$ which was one of the conditions on $n_l$.
Therefore,
$$
\limsup_{N\to \infty} \left(\sum_{j=1}^N j^{-\alpha} \,\sin(2\bar\theta_j)\right)\,/\,
\left(\sum_{j=1}^N j^{-\alpha} \right)\,\leq\, q_l^{-\alpha}\,\to\,0\qtx{for} q_l\,\to\,\infty\;.
$$
Repeating these arguments for the second and third term in \eqref{eq-oscil-terms} finishes the proof of \eqref{eq-log(T)}.

For $\alpha<1$, \eqref{eq-exist-subord} follows directly from \eqref{eq-log(T)} and \cite[Theorem~8.3]{LaSi}. 
In the case $\alpha=1$, i.e. $d=2$, we need to use \cite[Lemma~8.7]{KLS} and verify the conditions following the arguments of
the proof of \cite[Lemma~8.8]{KLS}.

So let $\alpha=1,\,d=2$ and $\beta=c\, \gamma_{E,\lambda}$, then we have $\log\|T_E(n)\|\,/\,\log(n)\,\to\,\beta$. Moreover, let $R^{(j)}_n$ and $\bar\theta_n^{(j)}=\theta^{(j)}_n+k$ for $j=1,2$ be defined as $R_n$ and $\bar\theta_n=\theta_n+k$ in \eqref{eq-def-R-th} and \eqref{eq-rec-R} with starting angles $\theta^{(1)}_0=0$ and $\theta^{(2)}_0=\pi/2$.
As in \cite{KLS} we get almost surely $\log|\theta_n^{(1)}-\theta_n^{(2)}|\,/\,\log(n)\to-2\beta$.
Then, using $f(x,\bar\theta)$ as in \eqref{eq-def-f(x,th)}, define the random variable
$$
L(n)\,:=\,f(x_n,\bar\theta^{(1)}_n)\,-\,f(x_n,\bar\theta^{(2)}_n)\;.
$$
Following the proof of \cite[Lemma~8.8]{KLS} the main point is to show that for any $\varepsilon>0$ we have (cf. \cite[eq. (8.22)]{KLS})
\begin{equation}\label{eq-est-L(n)}
\left|\sum_{n=N}^\infty L(n)\right| \,\leq\, C_\omega N^{-2\beta+\varepsilon}
\end{equation}
almost surely, for some random variable $C_\omega$
(recall that $\omega\in\Omega$ denoted the randomness).
There are some differences in the setup here to arrive at this estimate.

Taking $J>2+2\beta$ and noting that $s_n\sim n$ in this case $(d=2)$, we have by \eqref{eq-x-bounds2} that $\EE(x_n^{2J})=O(n^{-J})=o(n^{-2-2\beta})$.
A Borel Cantelli argument as in Lemma~\ref{lem-conv} shows that
$x_n^{2J}=o(n^{-2\beta-1+\varepsilon})$ almost surely. Furthermore, we have $|\theta_n^{(1)}-\theta_n^{(2)} |=o(n^{-2\beta+\varepsilon})$ almost surely.
Taking a Taylor expansion of $f(x,\bar\theta)$ in the first variable up to reminder term $O(x_n^{2J})$ we see that
$\EE(L(n))=o(n^{-2\beta-1+\varepsilon})$ and
$$
L(n)\,=\,\sum_{j=1}^{2J-1} (x_n^j-\EE(x_n^j))\,\left[ f^{(j)}(0,\bar\theta^{(1)}_n)-f^{(j)}(0,\bar\theta^{(2)}_n)\right]\,+\,o(n^{-2\beta-1+\varepsilon})\;.
$$
$\bar\theta_n^{(1)},\,\bar\theta_n^{(2)}$ depend only on $x_0,\ldots,x_{n-1}$ and the variance of each term is of order
$o(n^{-4\beta-1+2\varepsilon})$. Therefore, we can use \cite[Lemma~8.4 part(3) with $2\alpha=1+4\beta-2\varepsilon$]{KLS} to obtain \eqref{eq-est-L(n)}.
After this estimate we can conclude as in \cite[Lemma~8.8]{KLS} and obtain that the assumptions for \cite[Lemma~8.7]{KLS} 
are fulfilled almost surely. This gives \eqref{eq-exist-subord} in the case $d=2$.


\end{document}